\documentclass[aps,prd,amsmath,amssymb,10pt,letterpaper,balancelastpage,showpacs,superscriptaddress,nofootinbib,floatfix,notitlepage,preprintnumbers,longbibliography,twocolumn]{revtex4-2}

\usepackage{graphicx}
\usepackage{soul}
\usepackage[caption=false]{subfig}
\usepackage{color}
\usepackage{xcolor}
\usepackage{enumitem}
\usepackage{cancel}
\usepackage{mathtools}
\usepackage{hyperref}
\usepackage{orcidlink}
\usepackage{cases}
\usepackage[normalem]{ulem}

\usepackage{amsthm}
\newtheorem{theorem}{Theorem}

\newtheorem{lemma}[theorem]{Lemma}
\theoremstyle{remark}
\newtheorem*{remark}{Remark}
\theoremstyle{remark}

\theoremstyle{plain}
\newtheorem{manuallemmainner}{Lemma}
\newenvironment{manuallemma}[1]{%
  \IfBlankTF{#1}
    {\renewcommand{\themanuallemmainner}{\unskip}}
    {\renewcommand\themanuallemmainner{#1}}%
  \manuallemmainner
}{\endmanuallemmainner}

\newtheorem{manualtheoreminner}{Theorem}
\newenvironment{manualtheorem}[1]{%
  \IfBlankTF{#1}
    {\renewcommand{\themanualtheoreminner}{\unskip}}
    {\renewcommand\themanualtheoreminner{#1}}%
  \manualtheoreminner
}{\endmanualtheoreminner}

\usepackage{tabularray}
\UseTblrLibrary{booktabs} 

\newcommand{\pfrac}[2]{{\left(\frac{#1}{#2}\right)}}

\newcommand{\MAX}{{\rm (max)}}

\newcommand{\avg}{{\rm avg}}
\newcommand{\tot}{{\rm tot}}
\newcommand{\Var}{{\rm Var}}

\newcommand{\bra}[1]{\left\langle #1 \right|}
\newcommand{\ket}[1]{\left|#1\right\rangle}
\newcommand{\braket}[2]{\left\langle#1 |  #2\right\rangle} \newcommand{\expect}[1]{\left\langle#1\right\rangle}

\newcommand{\Tr}{{\rm Tr}}
\newcommand{\nnorm}[1]{\left|\left|#1\right|\right|}

\newcommand{\EMF}{E_{\rm MF}}

\DeclareMathOperator*{\Avg}{avg}

\newcommand{\TrD}[2]{\frac{1}{2}\left|\left|#1 - #2\right|\right|_1}

\definecolor{figureRed}{RGB}{255, 127, 127}
\definecolor{figureGreen}{RGB}{127, 215, 167}
\definecolor{figureBlue}{RGB}{127, 183, 223}
\newcommand{\legsq}[1]{\textcolor{#1}{\rule{1.5ex}{1.5ex}}}

\begin{document}

\preprint{FERMILAB-PUB-26-0320-SQMS-T}

\title{Entanglement Requirements for Coherent Enhancement in Detectors}
\date{\today}

\author{Zachary Bogorad\,\orcidlink{0000-0001-9913-6474}\,}
\email{zbogorad@fnal.gov}
\affiliation{Fermi National Accelerator Laboratory, Batavia, IL 60510, USA}

\author{Roni Harnik\,\orcidlink{0000-0001-7293-7175}\,}
\email{roni@fnal.gov}
\affiliation{Fermi National Accelerator Laboratory, Batavia, IL 60510, USA}

\begin{abstract}
    \noindent Coherent enhancement is a powerful mechanism for improving the sensitivity of a wide range of detectors, but its practical use is often limited by the difficulty of preparing the required quantum states. We show that this difficulty has a fundamental origin: coherent enhancement of a signal interacting with a detector is quantitatively constrained by entanglement. We prove general bounds on how the strength of coherent effects can scale with system size, as a function of the single-mode entanglement entropy of the detector. These bounds smoothly interpolate between the incoherent and fully coherent regimes, and apply both to parameter-estimation problems and to scattering processes. We discuss these results from two complementary perspectives: First, they appear as bounds on the quantum Fisher information of many-body states, which translate directly into limits on parameter sensitivity via the quantum Cram{\'e}r-Rao bound. Second, they can be interpreted as limits on a class of scattering cross sections, leading to predictions for how minimum detectable interaction strengths scale with target size. Together, these results provide a unified view of coherent enhancement in metrology and scattering experiments, and motivate the development of new techniques for generating entangled detector states.
\end{abstract}

\maketitle


Coherent enhancement—the constructive addition of amplitudes from many constituents of a detector or target—underlies some of the most powerful scaling advantages in quantum physics. Well-known examples include Heisenberg-limited parameter estimation in quantum metrology \cite{Giovannetti:2011chh, InterferometerCoherenceReview, Toth:2014msl, Degen:2016pxo}, collective emission phenomena \cite{Gross:1982dkt, Masson:2021eyk, Asenjo-Garcia:2017zjf, Andreev_1980}, and coherent scattering from, or absorption by, composite targets \cite{Freedman:1973yd, Akimov:2017ade, Baxter:2019mcx, Jeffries2021, Sivia2011, Chen:2023swh, Ito:2023zhp, GHZScattering, Bogorad:2023zmy}. In each of these settings, coherence can in principle convert an incoherent linear scaling with system size into a quadratic one. At the same time, it is broadly appreciated that such enhancements are difficult to realize in practice, and are often associated with the preparation of highly nonclassical many-body states.

In this work, we show that coherent enhancement is quantitatively limited by entanglement. We derive general bounds on collective responses generated by sums of local operators, demonstrating that the coherent contribution is controlled by the correlations present in the underlying many-body state. For pure states, these bounds can be expressed in terms of the average single-subsystem von Neumann entropy, yielding a simple interpolation between incoherent $O(N)$ scaling and maximally coherent $O(N^2)$ scaling; mixed states admit similar bounds via an entanglement of formation-type quantity \cite{Bennett:1996gf, Wootters:1997id, Horodecki:2009zz, Plenio:2007zz}. States with intermediate entanglement exhibit intermediate scaling, placing universal limits on how much coherence can be extracted from partially entangled detectors or targets. The analysis is largely model-independent and relies only on the structure of the detector or target Hilbert space and the entanglement properties of its quantum state. 

Our results provide a unified information-theoretic framework for understanding the limits of coherent enhancement across quantum sensing, metrology, and scattering experiments, and clarify the entanglement resources required to surpass incoherent scaling. They refine the relationship between multipartite entanglement and enhanced quantum Fisher information, extending  the Cram{\'e}r-Rao bound \cite{InterferometerCoherenceReview, Paris:2008zgg, Liu:2019xfr, Helstrom:1969fri} and the familiar distinction between the standard quantum limit and Heisenberg-limited scaling to a broader class of states. The analogous scattering bounds constrain coherent contributions to inelastic transition probabilities and thereby limit how detectable interaction strengths can scale with target size. This perspective is particularly relevant for particle physics experiments that rely on coherent scattering to enhance sensitivity to weakly coupled probes \cite{Freedman:1973yd, Akimov:2017ade, Baxter:2019mcx, Anderson:2011bi, Dutta:2024kuj, DeRomeri:2023dko, SuperradiantDetection1, SuperradiantDetection2, SuperradiantDetection3, CoherentDecoherence1, CoherentDecoherence2, Du:2022ceh, Brady:2022bus, Brady:2022qne}, but the framework itself is independent of any specific physical realization.

The remainder of this work is organized as follows: We first briefly discuss a toy example motivating a connection between coherence and entanglement. The main text then asserts simplified versions of our key results, omitting their proofs (which are deferred to Appendix~\ref{app:TechnicalDetails}). This allows us to present our core ideas while minimizing burdensome notation. Finally, we discuss a series of examples illustrating the various behaviors---incoherent, coherent, and in between---of our results in both the scattering and quantum sensing contexts. The Supplemental Materials include a table of symbols, our generalized results with proofs, some additional technical details, and further discussion of the intuition underlying our results.

\begin{figure*}[t]
    \subfloat{%
      \includegraphics[width=0.25\textwidth]{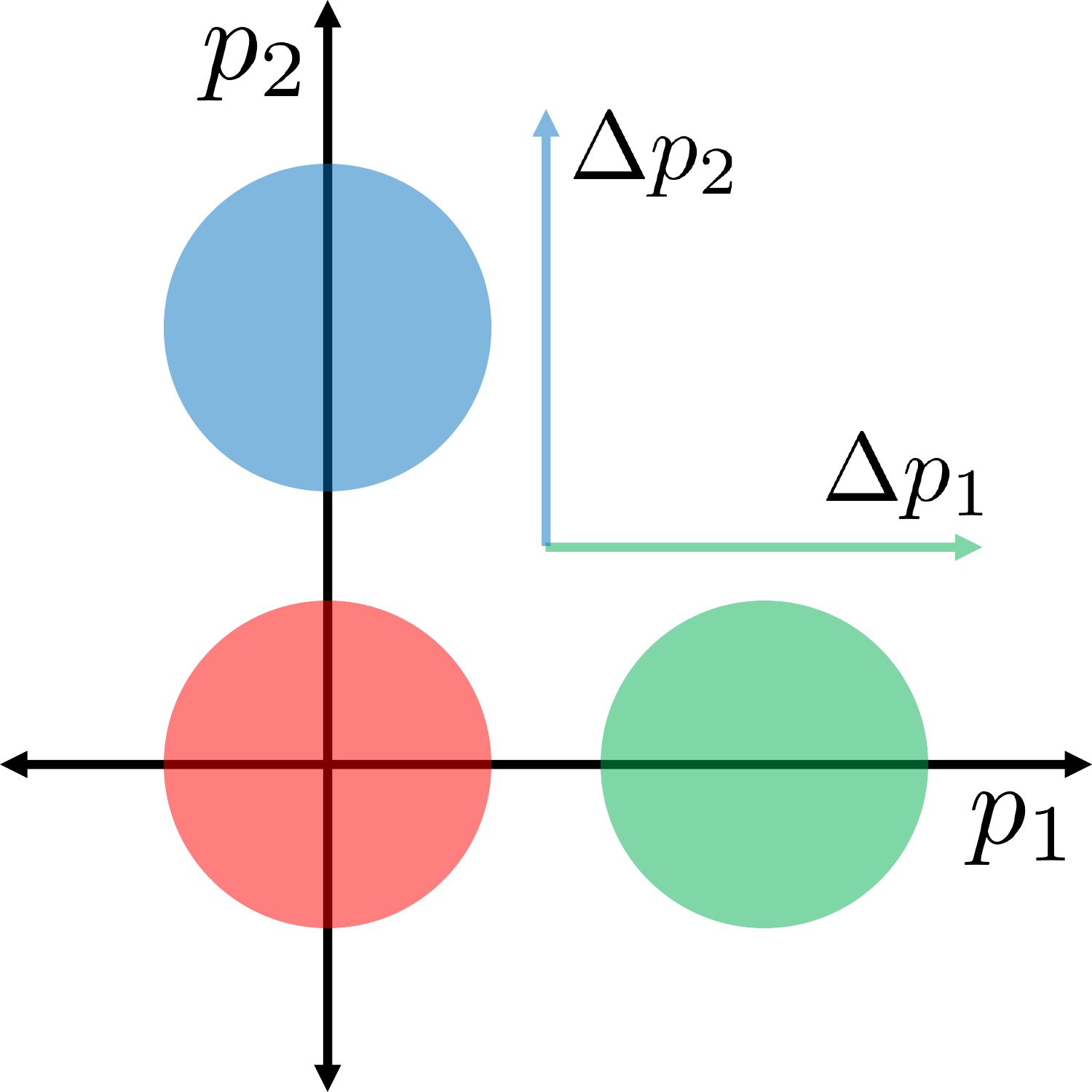}
    }
    \hspace{5mm} 
    \subfloat{%
      \includegraphics[width=0.25\textwidth]{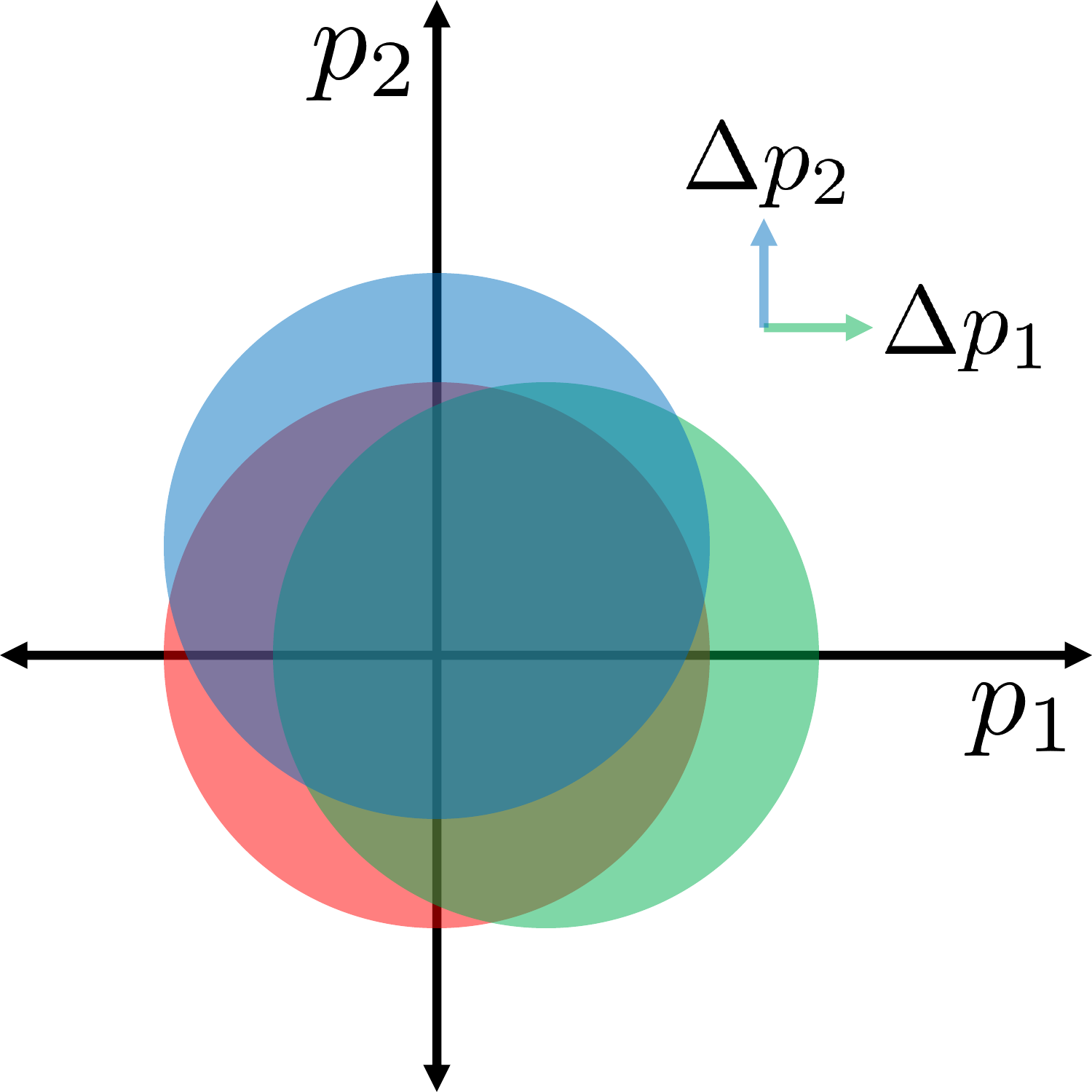}
    }
    \hspace{5mm} 
    \subfloat{%
      \includegraphics[width=0.25\textwidth]{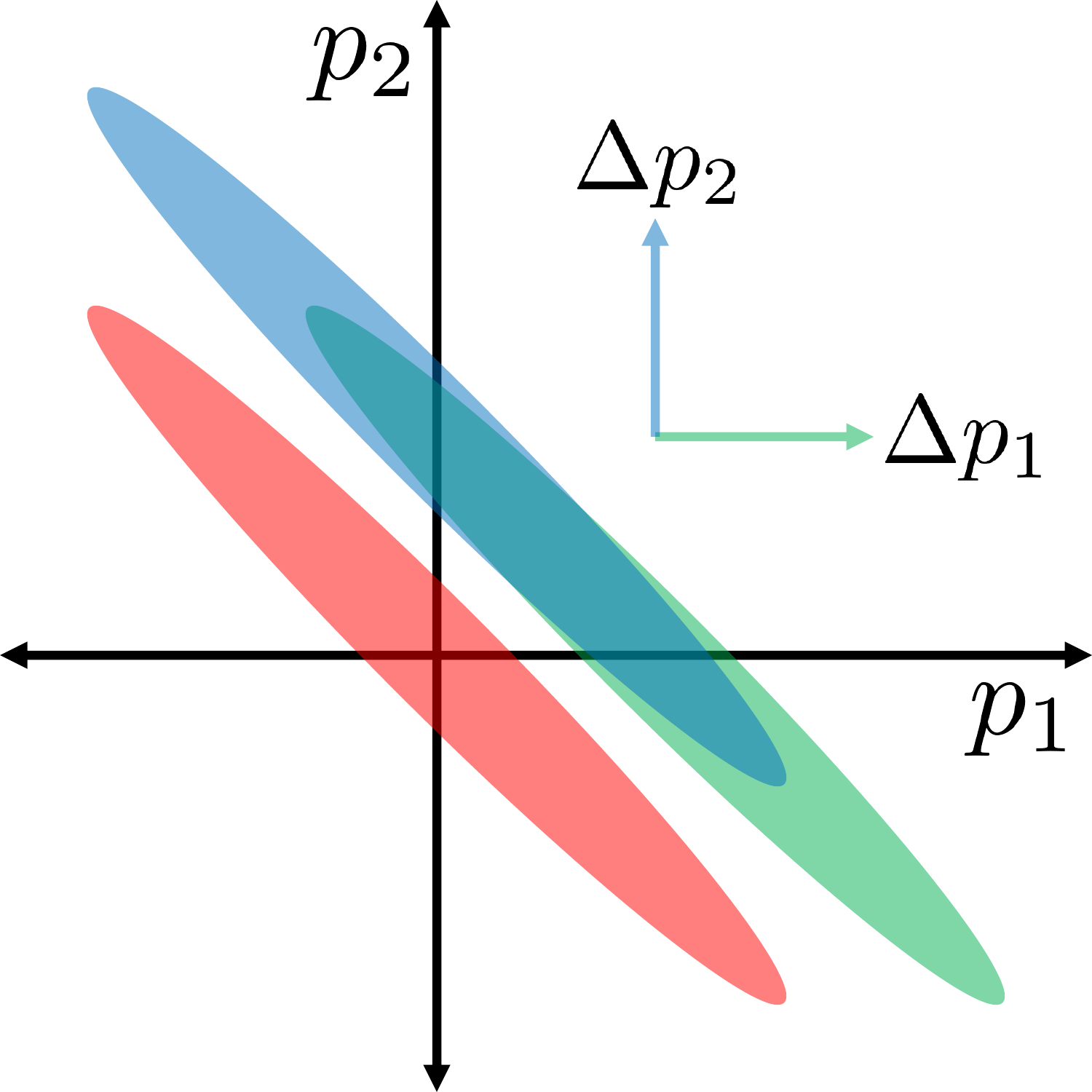}
    }
    \caption{Three examples of initial and final target particle state momentum distributions for a 1D scattering process at a fixed momentum transfer from two target particles; see text. In each plot, the red region shows the initial target state while the green and blue regions show the final state contributions from scattering off of the first or second particle, respectively; the true final state will thus be a superposition of the green and blue regions, assuming no decoherence or relative phase. These illustrate, from left to right: incoherent but easily detectable, coherent but poorly detectable, and coherent detectable scattering.}
    \label{fig:MotivatingExample}
\end{figure*}

\textbf{Motivating Example} --- We begin by presenting an illustrative, qualitative example of why one should expect coherent enhancements in detectors to be related to entanglement at all\footnote{We discuss this example, and a number of others, in further detail in Appendix~\ref{app:IntuitivePictures}.}. Consider a scattering process in which an incident particle scatters from a target\footnote{We emphasize that we use the term ``target'' in the direct detection sense---where a change in the target state is the basis for signal readout---and not in the fixed-target experiment sense, where detection happens separately from the target. Our results can apply to fixed-target experiments, but they apply to their detectors, not their ``targets.''} consisting of two particles, all treated in 1D, and choose a particular momentum transfer.

There are then two possible scattering processes at leading order: the momentum can be transferred to the first particle, or to the second. These two possibilities would lead, individually, to distinct final states. However, if there is no relative phase between these two outcomes---because the spatial separation between the two particles is negligible, for example---the true final state of this scattering process will instead be a superposition of the two states. Coherent enhancement can be understood as constructive interference between them.

Coherence is not all we desire in a detector, however: the final state of the scattering process must be significantly different from the initial state, in order for us to be able to see that a scattering event happened.

Now consider three possible initial states, shown in red in the panels of Figure~\ref{fig:MotivatingExample}. If the initial momentum distributions of the two target particles are independent, and both have uncertainties small compared to the momentum transfer (the left panel), the two final state contributions occupy disjoint or nearly disjoint regions of momentum space, and thus constructive interference is absent. The final states are also disjoint from the initial state, however. This process is therefore detectable, but not coherently enhanced.

The central panel considers the opposite limit of independent distributions: large momentum uncertainties compared to the momentum transfer. These result in initial and final states that all heavily overlap; although this can lead to a coherent enhancement in the cross section, the states are difficult to distinguish, reducing the effectiveness of this initial state (although, as we will show, detectors in this regime can still perform parametrically better than those in the previous one).

To achieve simultaneous coherence and high detectability, one thus needs the initial state to have correlated momenta; in particular, we require entanglement. We demonstrate how this can result in superior sensitivity in the right panel, which shows how an appropriately-entangled state (e.g. a bound state) can give final state contributions that overlap with each other but not with the initial state.

The remainder of this work is devoted to quantifying this requirement, and in particular to smoothly connecting the unentangled incoherent regime to the entangled coherent one.

\textbf{Mathematical Structure} --- The results of this paper are all consequences of the following linear algebra statement. Despite its technical appearance, we show in the remainder of this work that it is immediately applicable to a wide range of physically-interesting scenarios, bounding the sizes of observable coherent effects via the entanglement entropy of the components of a detector or experiment.

This central result involves four objects that will appear repeatedly in this work: a Hilbert space
\begin{align}
    \mathcal{H} = \bigotimes\limits_{k=1}^N \mathcal{H}_k \label{eq:Hdefinition}
\end{align}
where each $\mathcal{H}_k$ is itself a finite-dimensional\footnote{Although many physical systems to which our results could be applied are nominally infinite-dimensional, this can often be remedied by IR and UV cutoffs---most typically a finite spatial extent and an energy cutoff.} Hilbert space; an operator
\begin{align}
    G = \sum\limits_{k=1}^N G_k \bigotimes\limits_{\ell \neq k} I_\ell \label{eq:Gdefinition}
\end{align}
where each $G_k$ is a generic linear operator (not necessarily Hermitian) on $\mathcal{H}_k$ and $I_\ell$ is the identity on $\mathcal{H}_\ell$; a density matrix $\rho$ on $\mathcal{H}$, which will later correspond to the initial states of targets and detectors; and a von Neumann entropy for each individual subspace (e.g. for individual particles or modes within the detector)
\begin{align}
    s_k = -\Tr\left[ \rho_k \ln \rho_k \right] \label{eq:skDefinition}
\end{align}
where $\rho_k$ is the reduced density matrix on $\mathcal{H}_k$ alone (i.e. after tracing over everything except $\mathcal{H}_k$). For pure states~$\rho$, $s_k$ is a measure of the entanglement of the $k$'th particle or mode with the rest of~$\rho$. This is no longer true when~$\rho$ is a mixed state, in which case one must employ more involved measures of~$\rho$'s entanglement; we discuss this in Appendix~\ref{app:TechnicalDetails}.

We make two simplifying assumptions in the main text: First, we assume that every $\mathcal{H}_k$, every $G_k$ (as well as the corresponding operators in later theorems, i.e. $H_k$ and $T_k$), and every $s_k$ is identical. Second, we assume that all states $\rho$ are pure\footnote{This is not strictly necessary---all of our results still hold for mixed states---but mixed states can have non-zero $s_k$ even in the absence of internal entanglement, due to entanglement with the environment. The results in Appendix~\ref{app:TechnicalDetails}, which exclude this environmental contribution, are therefore both stronger and more physically meaningful for mixed states.}. This renders the main text results shorter and more readable, while making only an order-one difference for typical applications of our work, where most states are likely to be pure and each of these objects is likely to have little or no variation with $k$. We emphasize however that neither of these assumptions is required; accounting for each is simply a matter of taking appropriate averages, which we do for each of our results in Appendix~\ref{app:TechnicalDetails} (e.g. Lemma~\ref{lem:VarianceBound} below is generalized as Lemma~\ref{lem:VarianceBoundFull} in that appendix).

We also caution that, although our results all formally hold for both distinguishable and indistinguishable particles, as well as both bosons and fermions, their interpretation for indistinguishable bosons can become somewhat unintuitive; we discuss this in Appendix~\ref{app:IndistinguishableBosons}. In particular, for collections of photons (e.g. in electromagnetic cavities), the forms above typically hold with $N$ the number of distinct modes, rather than the number of photons.

\begin{lemma}
    For any Hilbert space $\mathcal{H}$ and operator $G$ of the forms in Eqs.~\eqref{eq:Hdefinition} and~\eqref{eq:Gdefinition}, and any state $\rho$ on $\mathcal{H}$, the variance of $G$ is bounded by
    \begin{align}
    \begin{split}
        & {\rm Var}(G) \equiv \Tr[\rho G^\dagger G] - \left| \Tr[\rho G] \right|^2 \\
        &~ \leq ||G_1||_\infty^2 \left( N + 2 N^2 \sqrt{s_1} \right),
    \end{split}
    \end{align}
    with $||G_1||_\infty$ the Schatten $\infty$-norm, i.e. the largest eigenvalue of $\sqrt{G_1^\dagger G_1}$.\footnote{Here, as in the remainder of the main text, we single out $G_1$ (and $s_1$, etc.) because we are assuming $k$-independent $G_k$, etc.; the generalized results are in Appendix \ref{app:TechnicalDetails}.}
    \label{lem:VarianceBound}
\end{lemma}

Two limits of this result are easily understandable: when there is no entanglement between subspaces ($s_\avg = 0$), there can be no covariance between distinct $\mathcal{H}_k$ and thus the total variance is necessarily linear in $N$. A sufficiently entangled system, however, can of course regain the generic $N^2$ scaling of variance. Likewise, the prefactor of $||G_1||_\infty^2$ is straightforwardly interpreted as reflecting the maximum achievable variances of the individual operators. The non-trivial implication of Lemma~\ref{lem:VarianceBound} is the scaling of the intermediate regime, where the maximum variance scales as the square root of the average single-subsystem von Neumann entropy. 

Lemma~\ref{lem:VarianceBound} is, however, only an upper bound, and in fact we are not aware of any systems exceeding $O(N + N^2 s_1)$ rather than $O(N + N^2\sqrt{s_1})$; see the examples below. We leave proving or disproving such an improved bound to future work, however. Our bound becomes particularly loose when systems are highly inhomogeneous, e.g. due to entanglement that is local (between each subsystem and only a few others) rather than widely distributed; we discuss this further in Appendix~\ref{subapp:VarianceProof}.

We now turn to the physics applications of this result, beginning with a bound on quantum parameter estimation and then turning to a bound on scattering sensitivity. Although closely related, these physical results will apply to different classes of experiments, with different resulting scaling behavior.

\textbf{Quantum Fisher Information and the Cram{\'e}r-Rao Bound} --- The most immediate application of Lemma \ref{lem:VarianceBound} is to the quantum Fisher information (QFI), which quantifies how rapidly a state changes under a particular Hamiltonian. Since the QFI is upper bounded by 4 times the variance of the Hamiltonian (with equality for pure states) \cite{QFIReview}, we can apply Lemma~\ref{lem:VarianceBound} and rearrange the result to show the following:

\begin{lemma}
    The QFI of a state $\rho$, given any Hamiltonian of the form in Eq.~\eqref{eq:Gdefinition}, is bounded by
    \begin{align}
        \frac{ F_Q[\rho, H] }{F_Q^\MAX[H_1]} \leq N + 2N^2\sqrt{s_1}
    \end{align}
    where $F_Q^\MAX[H_1]$ is the maximum QFI that could be achieved with any state on $\mathcal{H}_1$ alone.
    \label{lem:QFIBound}
\end{lemma}

This can immediately be substituted into the quantum Cram{\'e}r-Rao bound, a well-known limit on the best possible estimate of a Hamiltonian parameter using a given state as a function of the quantum Fisher information (see e.g. Ref.~\cite{QFIReview}). This allows us to bound how well such a measurement can scale with system size\footnote{See also Ref.~\cite{Gorecki:2024rqd} for other bounds based on the QFI.}:

\begin{theorem}
     Given the unitary evolution of a state $\rho$ to $e^{-i\theta H}\rho e^{i\theta H}$, the best precision with which one can determine $\theta$ in $m$ measurements is
    \begin{align}
        (\Delta \theta)^2 \geq \frac{1}{m F_Q^\MAX[H_1] \left(N + 2N^2\sqrt{s_1}\right)}. \label{eq:NParticleCR}
    \end{align}
    \label{thm:CramerRaoEntropyDensity}
\end{theorem}

Omitting the $N+2N^2\sqrt{s_1}$ factor, Eq.~\eqref{eq:NParticleCR} would simply be Cram{\'e}r-Rao bound for the best possible state of a single subsystem $\mathcal{H}_1$. The $N$-dependent factor then gives the best possible improvement on that bound as the detector grows to $N$ particles or modes.

Theorem~\ref{thm:CramerRaoEntropyDensity} is, in effect, a quantitative formalization of the standard intuition that entanglement is necessary to beat the standard quantum limit. It is well known that, in the absence of any entanglement, the error on $\theta$ can scale at best as $N^{-1/2}$, i.e. incoherently (although the definition of $N$ here can be subtle; see Appendix~\ref{app:IndistinguishableBosons}). Conversely, fully coherent improvement of the error to $O(N^{-1})$ requires non-zero, $N$-independent entropy density\footnote{Analogous results are also known for other quantities, e.g. the power absorbed in a cavity, which can grow super-extensively only in the presence of correlations between modes \cite{Lasenby:2019hfz}.}. Theorem~\ref{thm:CramerRaoEntropyDensity} smoothly connects these limits, giving a minimum entropy scaling required to achieve a particular scaling of sensitivity with $N$.

\textbf{Scattering} --- Translating Lemma~\ref{lem:VarianceBound} into bounds on scattering processes is only slightly more involved. We begin by considering a single particle incident on and/or outgoing from an $N$-particle target, encompassing both scattering proper as well as emission and absorption processes; we refer to all three as ``scattering'' for brevity. 

We are interested in scattering processes which, at leading order, correspond to an interaction between the incident/outgoing particle and a single target particle. We can then write the S-matrix as
\begin{align}
    S = 1 + i\alpha \sum\limits_{k=1}^N T_k^{(1)} + O(\alpha^2), \label{eq:SMatrixSeparation}
\end{align}
where $\alpha$ is some perturbative expansion parameter and each $T_k^{(1)}$ acts only on the scattered particle and on $\mathcal{H}_k$ (i.e. it acts as the identity on every other $\mathcal{H}_\ell$).  Note that, for a coupling $g$ appearing in the Lagrangian, $\alpha$ may be order $g$ (for absorption/emission processes), order $g^2$ (for scattering processes), or even other combinations of one or more couplings. $\alpha$ should always be understood as defined by Eq.~\eqref{eq:SMatrixSeparation}.

We begin by bounding a quantity we will call the ``orthogonal scattering probability,'' $p_\perp$: the scattering probability\footnote{For emission processes, this must necessarily be understood as a scattering probability per unit time, with the appropriate dimensionful parameters put in. This does not substantively change our results, however.} restricted to final states of the target orthogonal to the initial target state\footnote{One could of course rewrite our bounds in terms of scattering cross sections rather than probabilities if preferred, but this would prove to be less convenient below.}. For an initial (pure) target state $\ket{t}$, this is
\begin{align}
    p_\perp = \bra{i}\bra{t} S^\dagger \Pi_\perp S \ket{i}\ket{t} \label{eq:pPerpDefinition}
\end{align}
where $\ket{i}$ is the state of the incident particle and $\Pi_\perp = 1-\ket{t}\bra{t}$ is the projector onto the space orthogonal to $\ket{t}$. We will sometimes further restrict to the probability for a particular outgoing state $\ket{f}$, $p_\perp^{i \to f}$, simply by multiplying the projector by $\ket{f}\bra{f}$.

The orthogonal scattering probability has a convenient property: it is determined to order $\alpha^2$ by only the order-$\alpha$ terms in $S$. To see this, note that, on substituting Eq.~\eqref{eq:SMatrixSeparation} into Eq.~\eqref{eq:pPerpDefinition}, the $O(\alpha^2) \times O(1)$ cross-terms are eliminated by $\Pi_\perp$, leaving
\begin{align}
    p_\perp = \alpha^2 \sum\limits_{k, \ell} \bra{i}\bra{t} T_k^{(1)\dagger} \Pi_\perp T_\ell^{(1)} \ket{i}\ket{t} + O(\alpha^3). \label{eq:pPerpSimplified}
\end{align}
We will thus be able to study observables at order $\alpha^2$ below without specifying the order-$\alpha^2$ behavior of $S$.

\begin{lemma}
    The total orthogonal scattering probability (satisfying some technical conditions given in Appendix~\ref{subapp:ScatteringProof}) from an $N$-particle/mode pure target state $\rho$ at order $\alpha^2$ is upper bounded by
    \begin{align}
        p_\perp &\leq p_{\perp,1}^{\rm max} \left( N + 2 N^2 \sqrt{s_1} \right), \label{eq:ScatteringProbBound}
    \end{align}
    where $p_{\perp,k}^{\rm max}$ is the maximum orthogonal scattering probability from any state of the $k$'th particle/mode alone.
    \label{lem:OrthogonalScatterProbTot}
\end{lemma}

Lemma~\ref{lem:OrthogonalScatterProbTot} bounds only the orthogonal scattering probability $p_\perp$, while some scattering experiments perform measurements over bases that include superpositions of the initial target state with other states (see e.g. Refs. \cite{CoherentDecoherence1, CoherentDecoherence2}). Nonetheless our bound suffices to constrain the sensitivity of \textit{any} scattering experiment: as we show in Lemmas~\ref{lem:TraceDistanceBound} and~\ref{lem:SingleParticleGeneralLimit}, $p_\perp$ bounds the trace distance between the final target states in the presence and absence of scattering, which in turn bounds the achievable statistical power of any positive operator-valued measure (POVM). Measurements in bases that include superpositions of the initial target state may yield large outcome probabilities even in the absence of entanglement, but the difference in those probabilities between the scattering and no-scattering hypotheses is controlled by the same trace distance bound.

One additional subtlety arises in making this generalization, however: a distinction arises between fully forward elastic scattering processes, in which the incident and outgoing particle states are exactly the same (e.g. in the Stodolsky effect~\cite{Stodolsky:1974aq}), and all other scattering processes. In particular, although the former can change the state of the target at $O(\alpha)$, the latter changes it only at $O(\alpha^2)$ after integrating out the unknown outgoing particle state; this will lead to parametrically different sensitivity in our later results. Note that fully forward elastic scattering is nothing more than the unitary evolution measured in a QFI/Cram{\'e}r-Rao context described in a different language---tracing out the incident/outgoing state gives an effective Hamiltonian for the target alone---and thus the $O(\alpha)$ results will precisely mirror Lemma~\ref{lem:QFIBound} and Theorem~\ref{thm:CramerRaoEntropyDensity}.

\begin{lemma}
    For any scattering process of the same form as in Lemma~\ref{lem:OrthogonalScatterProbTot}, the trace distance between the final states in the presence ($\rho_f'$) and absence ($\rho_f$) of the incident particle is upper bounded by
    \begin{align}
        \TrD{\rho_f'}{\rho_f} \leq \sqrt{p_\perp^{i\to i}} + \sqrt{p_\perp p_s} \label{eq:TraceDistanceBound}
    \end{align}
    where $p_\perp^{i\to i}$ is the contribution to $p_\perp$ from $\ket{f} = \ket{i}$ (i.e. fully forward elastic scattering of the incident particle) and $p_s$ is the total probability of scattering to an outgoing particle state different from the incident particle state\footnote{Note that $p_s$ is equivalent to $p_\perp$ but with the target and scattered particle labels swapped.}.
    \label{lem:TraceDistanceBound}
\end{lemma}

We now translate this bound on the trace distance into a bound on how well an optimal measurement can detect the scattering process we are considering using Helstrom's theorem (see e.g. Ref. \cite{Wilde}), which connects trace distances to optimal measurement errors. Such a detection limit is naturally a function of one's error tolerances (e.g. a p-value in typical frequentist contexts), however. In our case it will be most natural to do this in terms of a power (the probability of a positive result conditional on there being an incident particle) and a type I error (the probability of a positive result in spite of there not being an incident particle), but we discuss how these can be translated into more typical frequentist or Bayesian quantities in Appendix~\ref{subapp:SingleParticleGeneralLimitProof}.

\begin{lemma}
    For any chosen type I error rate and power, no experiment searching for a single particle scattering at $O(\alpha)$ from an $N$-particle detector can have a detection limit better than
    \begin{align}
        \frac{1}{\alpha_{\rm min}^2} = O\left( N + N^2 \sqrt{s_1} \right).
    \end{align}
    Moreover, if a scattering process has $p_\perp^{i\to i} = 0$ (i.e. if fully forward elastic scattering has no effect on the target), then no experiment can detect an effect below $O(\alpha^2)$ and, at that order, the best possible scaling of the detection limit is\footnote{Although the $O(\alpha^2)$ result may appear ``better'' in that it has stronger scaling with $N$ at low entanglement, this is misleading: the detectable effect is still parametrically smaller than an $O(\alpha)$ effect, since we assume perturbativity and thus $\alpha^2 N^2 \ll 1$.}
    \begin{align}
        \frac{1}{\alpha_{\rm min}^2} = O\left( N^{3/2} + N^2 s_1^{1/4} \right).
    \end{align}
    \label{lem:SingleParticleGeneralLimit}
\end{lemma}

\begin{remark}
    Note that, whereas the previous two lemmas are not well-defined for mixed states, this result (and the next) do have mixed state generalizations, presented in Appendix~\ref{app:TechnicalDetails}.
\end{remark}

In a typical scattering experiment, one is concerned with detecting the presence of many incident particles, rather than a single one. Fortunately the triangle inequality for trace distances means our bound can immediately be extended to this case, so long as the different incident particles can be treated independently, i.e. as long as we can write the total S-matrix for $K$ incident particles as 
\begin{align}
    S = S_K S_{K-1} ... S_2 S_1.
\end{align}
Only one new complication arises in this case: the $s_k$ can change over the integration time of the experiment. (In fact, as we illustrate in Appendix~\ref{app:IntuitivePictures}, some degree of entanglement accumulation is quite normal in a scattering experiment.) Although we are not aware of any situations in which this entanglement accumulation is sufficient to parametrically improve an experiment's sensitivity, we leave proving (or disproving) the universality of this to future work, and instead conservatively present Theorem~\ref{thm:MultipleParticleLimit}.

\begin{theorem}
     For any chosen type I error rate and power, no experiment searching for $K$ independent incident particles scattering at $O(\alpha)$ from an $N$-particle detector can have a detection limit better than
    \begin{align}
        \frac{1}{\alpha_{\rm min}^2} = O\left( K^2 N + K^2 N^2 \sqrt{s_1^\MAX} \right),
    \end{align}
    with $s_1^\MAX$ the largest value of $s_1$ achieved over the integration time. Moreover, if a scattering process has $p_\perp^{i\to i} = 0$ (i.e. if fully forward elastic scattering has no effect on the target), then no experiment can detect an effect below $O(\alpha^2)$ and, at that order, the best possible scaling of the detection limit is
    \begin{align}
        \frac{1}{\alpha_{\rm min}^2} = O\left( K N^{3/2} + K N^2 \left(s_1^\MAX\right)^{1/4} \right).
    \end{align}
     \label{thm:MultipleParticleLimit}
\end{theorem}

\begin{remark}
    In many contexts, $K$ is directly proportional to an integration time $\tau$, and could therefore be replaced by $\tau$ in this bound.
\end{remark}

\textbf{Examples} --- We now present a series of examples of both parameter estimation and scattering illustrating the various parametric behaviors in Theorems~\ref{thm:CramerRaoEntropyDensity} and~\ref{thm:MultipleParticleLimit}.

Incoherent $O(N)$ behavior is widespread and needs little explanation: scattering from disordered targets, as well as all sufficiently hard scattering, is well known to be incoherent (i.e. to have cross section scaling $\sigma \propto N$). Likewise, parameter estimation (e.g. magnetic field measurement) using a product state in, say, an atom interferometer gives sensitivity that scales with $N^{-1/2}$, as each atom contributes independently (see e.g. Ref. \cite{InterferometerCoherenceReview}).

Fully coherent $O(N^2)$ behavior is likewise familiar, although its connection to entanglement is not always appreciated. We begin with two situations in which coherent detection is clearly the result of entanglement, before turning to a more subtle example. In the parameter estimation picture, atom interferometry using maximally-squeezed states (i.e. Dicke states) leads to $O(N^{-1})$ sensitivity \cite{InterferometerCoherenceReview}; Dicke states, whose single-particle reduced density matrix is easily shown to be ${\rm diag}(1/2,1/2)$, have $s_k = \ln 2$, independent of $N$, as our bounds require. Similarly, scattering from GHZ-like states is known to scale as $\alpha_{\rm min} \sim N^{-1}$ in the noiseless limit\footnote{This can be parametrically suppressed in the presence of noise scaling with $N$, but even the noiseless case must satisfy our bounds.} \cite{GHZScattering}; the GHZ state again has $\rho_k = {\rm diag}(1/2,1/2)$ and thus $s_k = \ln 2$.

More traditional versions of $O(N^2)$ coherent scattering may, however, be less obviously consistent with our bounds. We take coherent elastic neutrino-nucleus scattering (CEvNS) as a concrete example, although the ideas generalize. In CEvNS, the neutrino-nucleus cross-section is known to scale with (approximately) the number of neutrons squared, i.e. coherently \cite{Freedman:1973yd, CEvNSIntro, Akimov:2017ade, Baxter:2019mcx}. Importantly, this enhanced cross-section is used in detection, so some $N$-independent fraction of it must be orthogonal (i.e. must contribute to $p_\perp$). Lemma~\ref{lem:OrthogonalScatterProbTot} thus requires that the nucleons be entangled. (In fact, this requirement can be intuitively understood; see Appendix~\ref{app:IntuitivePictures}.) Consistently with this, the nucleons within any nucleus must in fact be entangled with one another, simply by virtue of being in a bound state; an unentangled (i.e. product) state can have no correlations between particles. We discuss this further in Appendix~\ref{app:BoundStateEntanglement}.

Somewhat less familiar is the un-entangled coherent $O(N^{3/2})$ behavior specific to (non-fully forward elastic) scattering, as shown in Theorem~\ref{thm:MultipleParticleLimit}. Nonetheless it is not novel, and is exploited in, for example, Refs.~\cite{CoherentDecoherence1, CoherentDecoherence2, SuperradiantDetection1, SuperradiantDetection2, SuperradiantDetection3, Du:2022ceh}. As a simple example of this regime, consider a target composed of $N \gg 1$ spin-$1/2$ particles, initially prepared in the coherent half-up state $\ket{t} = 2^{-N/2}(\ket{\uparrow}+\ket{\downarrow})^{\otimes N}$, and an incident spin-up particle that can scatter, with negligible momentum transfer, by transferring one positive unit of spin to any target particle, i.e. by acting with $\ket{\uparrow}\bra{\downarrow}$ on any $\mathcal{H}_k$. In this case $\ket{t}$ is a product state, and thus has $s_k = 0$, so Lemmas~\ref{lem:OrthogonalScatterProbTot} and~\ref{lem:TraceDistanceBound} together allow a trace distance after one incident particle of $O(\alpha^2 N^{3/2})$. (This is equivalent to Lemma~\ref{lem:SingleParticleGeneralLimit}---see its proof---but in a more convenient form for this discussion.) The total scattering probability is easily shown to be $O(N^2)$; intuitively, this is possible because the final states for scattering from $\mathcal{H}_k$ and $\mathcal{H}_{\ell}$ are not orthogonal (see Appendix~\ref{app:IntuitivePictures}). However the final state $\ket{t'}$ after a scatter is not orthogonal to $\ket{t}$; a short calculation shows that $||\ket{t'}\bra{t'} - \ket{t}\bra{t}||_1/2 = 1/\sqrt{N+1}$, with zero trace distance if no scatter occurred. Thus the total trace distance is $||\rho_f'-\rho_f||_1/2 = O(\alpha^2 N^{3/2})$, precisely saturating our bound.

The $O(N)$, $O(N^2)$, and, for scattering, $O(N^{3/2})$ scalings are not the only achievable behaviors of a detector. As a simple parameter estimation example, consider again a target composed of $N \gg 1$ spin-$1/2$ particles, but now prepared in the not-quite-polar Dicke state
\begin{align}
    \ket{t} = \ket{J = \frac{N}{2},\, M = -\frac{N}{2} + \sqrt{N}}.
\end{align}
Here, $J$ is the total spin quantum number and $M$ is the quantum number for spin along the $z$ axis, i.e.
\begin{subequations}
\begin{align}
    \mathbf{J}^2 \ket{J, M} &= J(J+1) \ket{J, M} \\
    \mathbf{J}_z \ket{J, M} &= M \ket{J, M}.
\end{align}
\end{subequations}
where $\mathbf{J}_{x,y,z}$ is the total spin along a given axis, with $\mathbf{J}^2 = \mathbf{J}_x^2 + \mathbf{J}_y^2 + \mathbf{J}_z^2$. The reduced density matrix of a single spin in this state is $\rho_k = {\rm diag}(N^{-1/2}, 1-N^{-1/2})$, so we have
\begin{align}
    s_k = \frac{2+\ln N}{2\sqrt{N}} + O\left(\frac{1}{N}\right).
\end{align}
Since this is an eigenstate of $J_z$, we have $\expect{J_x} = 0$ while $\expect{J_x^2} = (J(J+1)-M^2)/2$ and thus
\begin{align}
    F_Q\left[ \rho, J_x \right] &= 2N^{3/2} + O(N).
\end{align}
This is a genuinely intermediate scaling behavior, between the incoherent $O(N)$ scaling and the fully coherent $O(N^2)$, but it does not saturate our bound (although it does, of course, satisfy it): Lemma~\ref{lem:QFIBound} gives $F_Q\left[ \rho, J_x \right] = O(N^{7/4} \sqrt{\ln N})$.

As an intermediate scattering example one can instead take $J=N/2$, $M = -N/2 + N^{2/3}$, with a scattering process that flips a single spin from down to up. This is easily shown to give $s_k \sim (3+\ln N)/(3N^{1/3})$ while $||\rho_f'-\rho_f||_1/2 \sim \alpha^2 N^{5/3}$, in spite of Lemmas~\ref{lem:OrthogonalScatterProbTot} and~\ref{lem:TraceDistanceBound} together allowing trace distances as large as $O(\alpha^2 N^{23/12} (\ln N)^{1/4})$; we again see that the bound holds but is somewhat loose.

In fact, as we noted earlier, we are not aware of any examples saturating the $O(N^2 \sqrt{s})$ operator variance bound for non-zero entanglement that goes to zero as $N\to\infty$, i.e. for $0  < s = o(1)$, but we leave proving a stronger bound (if one exists) to future work.

\textbf{Conclusion} --- In this work, we have derived general bounds on coherent enhancement in many-body detectors and targets, showing that superlinear scaling with system size is quantitatively constrained by entanglement. By bounding collective responses generated by sums of local operators, we demonstrated that their coherent contribution is limited by the correlations present in the underlying quantum state. For pure states, this leads to simple bounds in terms of the average single-subsystem entanglement entropy, while for mixed states the bounds can be expressed using a natural multipartite entanglement of formation–type quantity that excludes spurious enhancement from classical mixing.

These results provide a unified perspective on phenomena that are often treated separately, including enhanced quantum Fisher information in metrology and coherent contributions to scattering cross sections. In both cases, the familiar distinction between incoherent $O(N)$ scaling and maximally coherent $O(N^2)$ scaling emerges as a consequence of how entanglement is distributed across the system, with partially entangled states exhibiting intermediate behavior. Our analysis thus clarifies the entanglement resources required to surpass incoherent scaling, independent of the particular measurement system.

Beyond placing fundamental limits on achievable sensitivity, these bounds have practical implications for the design and interpretation of experiments that rely on coherent enhancement, ranging from atomic sensors to particle detectors. By framing coherent enhancement as a product of entanglement, our results motivate the development of new ways to produce and employ entanglement in order to improve measurements across a wide range of applications.

\acknowledgments

The authors would like to thank Paddy Fox, Peter Graham, Savas Dimopoulos, and Marios Galanis for insightful discussions. This  work was supported by the Office of High Energy Physics of the U.S. Department of Energy QuantISED program. 
RH is also partially supported by the Department of Energy, Office of Science, National Quantum Information Science Research Centers,
Superconducting Quantum Materials and Systems Center (SQMS) under contract number DE-AC02-07CH11359.
This manuscript has been authored  by  Fermi Forward Discovery Group, LLC under Contract No.~89243024CSC000002 with the U.S.~Department of Energy, Office of Science, Office of High Energy Physics.

\bibliographystyle{apsrev4-2}
\bibliography{main}

\clearpage
\newpage

\maketitle
\onecolumngrid
\makeatletter
\let\set@footnotewidth\set@footnotewidth@one
\let\compose@footnotes\compose@footnotes@one
\makeatother

\begin{center}
\textbf{\large Entanglement Requirements for  Coherent Enhancement in Detectors} \\ 
\vspace{0.05in}
{\it \large Supplemental Material}\\ 
\vspace{0.05in}
{Zachary Bogorad and Roni Harnik}
\end{center}

\setcounter{figure}{0}
\setcounter{table}{0}
\setcounter{section}{0}
\renewcommand{\thefigure}{S\arabic{figure}}
\renewcommand{\thetable}{S\arabic{table}}
\interfootnotelinepenalty=10000

\appendix

\section{Table of Symbols}\label{app:SymbolTable}

\vspace{-8mm}
\begin{longtblr}[
  label   = none, 
]{
  colsep = 0.008\textwidth,
  colspec = {Q[l,wd=0.08\textwidth] Q[l,wd=0.272\textwidth] Q[l,wd=0.6\textwidth]},
  rowhead = 1,
}
\toprule
Symbol & Expression & Meaning \\
\midrule
$N$ & & A number of subspaces that a detector/target can be factored into (e.g. a number of distinguishable particles, such that the detector/target can be written in terms of products of $N$ states). \\ \hline

$\mathcal{H}_k$ &  & A finite-dimensional Hilbert space (e.g. the Hilbert space of the $k$'th distinguishable particle in the detector/target). \\
$\mathcal{H}$ & $= \bigotimes\limits_{k=1}^N \mathcal{H}_k$ & The tensor product of the $N$ spaces $\mathcal{H}_k$. $\mathcal{H}$ is typically the Hilbert space of the detector/target we consider, without any incident or outgoing particles. \\ \hline

$I_k$ &  & The identity operator on $\mathcal{H}_k$. \\
$G_k$ &  & A generic (not necessarily Hermitian) linear operator on $\mathcal{H}_k$. \\ 
$G$ & $= \sum\limits_{k=1}^N G_k \bigotimes\limits_{\ell \neq k} I_\ell$ & The sum of the $N$ operators $G_k$, with each multiplied by identity operators on all other subspaces. \\ \hline

$H_k$ &  & A Hamiltonian (i.e. a Hermitian operator) on $\mathcal{H}_k$. \\
$H$ & $= \sum\limits_{k=1}^N H_k \bigotimes\limits_{\ell \neq k} I_\ell$ & The sum of the $N$ Hamiltonians $H_k$, with each multiplied by identity operators on all other subspaces. \\ \hline

$\rho$ &  & A density matrix on $\mathcal{H}$. Note that, unless otherwise stated, $\rho$ should be understood to include only the detector/target, and not any incident or outgoing particles. \\
$\rho_k$ & $=\Tr_{\mathcal{H}/ \mathcal{H}_k}\left[\rho\right]$ & The reduced density matrix of $\rho$ on subspace $\mathcal{H}_k$, i.e. after tracing over every  $\mathcal{H}_\ell$ for $k \neq \ell$. \\
$\rho_f$ $(\rho_f')$ & $=S \rho S^\dagger$ \qquad\qquad\qquad\qquad  (with $S$ restricted to the target Hilbert space; see text) & The density matrix of a target after an observation period in the absence (presence) of an incident particle that could scatter from it, or in the absence (presence) of an emitted particle. We work in the interaction picture, such that $S=1$ with no scattered particle. \\
\hline

$s_k$ & $= -\Tr\left[ \rho_k \ln \rho_k \right]$ & The single-subspace (e.g. single-particle or single-mode) von Neumann entropy of $\rho_k$. \\
$\EMF$ & See Equation \eqref{eq:EntanglementOfFormationDefinition}. & A multipartite entanglement of formation of the (possibly mixed) state $\rho$, i.e. the minimum probability-weighted average entropy per particle/mode that could be used to create $\rho$. \\\hline

$\theta$ & $\rho \to e^{-i\theta H}\rho e^{i\theta H}$ & A scaling factor for the evolution of a state $\rho$ via Hamiltonian $H$. \\
$\Delta\theta$ & & The uncertainty on $\theta$ following $m$ measurements using initial state $\rho$. \\
$m$ &  & The number of measurements performed. \\
$F_Q[\rho, H]$ & See e.g. Ref. \cite{QFIReview}. & The quantum Fisher information of the state $\rho$ with respect to Hamiltonian $H$. \\
$F_Q^\MAX[H]$ & $=\max\limits_{\rho} F_Q[\rho, H]$ & The maximum quantum Fisher information of any state $\rho$ with respect to Hamiltonian $H$. $\rho$ should be understood as being on the same Hilbert space as $H$, e.g. $\mathcal{H}_k$ for $F_Q^\MAX[H_k]$. \\ \hline

$S$ & $= 1 + i\alpha \sum\limits_{k=1}^N T_k^{(1)} + O(\alpha^2)$ & The S-matrix for a scattering process, e.g. $\ket{\psi_{\rm out}} = S\ket{\psi_{\rm in}}$ (see also $\rho_f$). (Here $\ket{\psi_{\rm in}} = \ket{i}\ket{t}$ and $\ket{\psi_{\rm out}} = \ket{f}\ket{t'}$.) \\
$T_k^{(1)}$ & See $S$. & The order-$\alpha$ $T$-matrix contribution of subspace $\mathcal{H}_k$ (e.g. the leading-order scattering contribution of the $k$'th particle). \\
$\alpha$ & See $S$. & An expansion parameter for the S-matrix (see $S$). Note that, given a coupling $g$, $\alpha$ may be order $g$ (for absorption/emission processes) or order $g^2$ (for scattering processes). \\
$\alpha_{\rm min}$ & & The smallest detectable value of $\alpha$ for a chosen experiment, given a desired power and type I error rate. \\
\hline

$\Pi_\perp$ & $=1-\ket{t}\bra{t}$ & The projector onto the space orthogonal to a pure state $\ket{t}$. \\
$p_\perp$ or $p_\perp^{i \to f}$ & $=\Tr\left[S\rho S^\dagger \Pi_\perp \right]$ \qquad\qquad\qquad  (with $S$ restricted to the target Hilbert space; see text) & The ``orthogonal scattering probability,'' i.e. the total probability to scatter into any target state orthogonal to the initial target state (defined only for an initially pure $\rho = \ket{t}\bra{t}$). May be restricted to a single final state of the outgoing particle ($p_\perp^{i \to f}$) or summed over all such possible states ($p_\perp$), changing the target space-restricted $S$ at left. Corresponds to a probability per unit time for emission/decay. \\
$p_{\perp,k}^{\rm max}$ & & The orthogonal scattering probability from the $k$'th particle or mode alone, maximized over states of that particle or mode. \\
$p_s$ & $=\Tr[S(\rho\otimes\ket{i}\bra{i})S^\dagger(1-\ket{i}\bra{i})]$ (using the full $S$) & The total (not just orthogonal) scattering probability for all outgoing states not equal to the incident state (i.e. summed over $\ket{f} \neq \ket{i}$). \\
\hline

$\ket{i}$ $(\ket{f})$ & & The incident (outgoing) particle state of the scattered, absorbed, or emitted particle. \\
$\ket{t}$ $(\ket{t'})$ & & The initial (final) state of the target. $\rho = \ket{t}\bra{t}$ for pure initial target states. \\ \hline

$\pi$ & & The power of a measurement, i.e. the probability of a positive signal when the searched-for event did, in fact, occur. \\
$\epsilon$ & & The type I error rate of a measurement, i.e. the probability of a positive signal when the searched-for event did not occur. \\ \hline

$||\ket{\psi}||$ & $=\sqrt{\braket{\psi}{\psi}}$ & The norm of a pure state. \\
$||M||_p$ & $=\Tr\left[|M|^p\right]^{1/p}$ & The Schatten $p$-norm of an operator $M$, where $|M| = \sqrt{M^\dagger M}$. For $p=\infty$, this is the largest eigenvalue of $\sqrt{M^\dagger M}$.\\
\SetCell[c=2]{l} $\frac{1}{2}||\rho-\rho'||_1$ & ~ & The trace distance between density matrices $\rho$ and $\rho'$, which determines how well an optimal measurement could distinguish the two states through e.g. Helstrom's Theorem (see e.g. Ref. \cite{Wilde}). \\
$D(\rho || \sigma)$ & $= \Tr\left[ \rho\ln\rho - \rho\ln\sigma \right]$ & The quantum relative entropy, a measure of the distinguishability of $\rho$ and $\sigma$ by extending the classical Kullback-Leibler divergence. \\
\bottomrule
\end{longtblr}

\section{Generalized Statements and Complete Proofs}\label{app:TechnicalDetails}

In this appendix, we generalize our main-text results to account for both particle/mode-dependent properties (different $\mathcal{H}_k$, different $G_k$, etc.) and for mixed states (since the main text results, while true for mixed states, are not well-suited to them). We also provide complete proofs of these generalized results, which immediately imply the main text results.

As the von Neumann entanglement entropy is not a good measure of entanglement for mixed states---notably, it can be non-zero for mixtures of unentangled states---some of the results in this section will instead use a multipartite entanglement of formation \cite{Bennett:1996gf, Wootters:1997id, Horodecki:2009zz, Plenio:2007zz},
\begin{align}
    \EMF &= \min\limits_{\rho = \sum_\ell p_\ell \ket{\psi_\ell}\bra{\psi_\ell}} \sum\limits_\ell p_\ell \Avg_k\left(s_{k,\ell}\right) \label{eq:EntanglementOfFormationDefinition}
\end{align}
where $\Avg_k(\cdot)$ is the average over values of $k$, the minimum is over all weighted mixtures of pure states that give the density matrix $\rho$ and $s_{k,\ell}$ is $s_k$ for the state $\ket{\psi_\ell}\bra{\psi_\ell}$. Intuitively, $\EMF$ is the minimum probability-weighted average single-particle entanglement entropy (between each particle/mode and the rest of the system) over all pure-state decompositions of $\rho$, and thus the results below will give a lower bound for the amount of entanglement that would need to be produced and maintained in an experiment in order to achieve a certain level of coherence. Note that, for pure states, $\EMF = \Avg_k(s_k)$, and, for the $k$-independent case in the main text, $\EMF$ is simply $s_k$.

\subsection{Lemma \ref{lem:VarianceBound}} \label{subapp:VarianceProof}

\begin{manuallemma}{1a}
    For any $\mathcal{H}$ and $G$ of the forms in Eqs.~\eqref{eq:Hdefinition} and~\eqref{eq:Gdefinition}, and any density matrix $\rho$ on $\mathcal{H}$,
    \begin{align}
    \begin{split}
        {\rm Var}(G) \equiv \Tr[\rho G^\dagger G] - \left| \Tr[\rho G] \right|^2 \leq \Avg_k \left(||G_k||_\infty^2 \right) \left( N + 2 N^2 \sqrt{ \Avg\limits_k\left(s_k\right) } \right). \label{eq:VarianceBoundFull}
    \end{split}
    \end{align}
    \label{lem:VarianceBoundFull}
\end{manuallemma}

\begin{proof}
    First expand the variance into pairwise terms,
    \begin{align}
        \Var(G) &= \sum\limits_{k=1}^N \Var(G_k) + \sum\limits_{k \neq \ell} \Tr\left[ G_k^\dagger G_\ell (\rho_{k\ell} - \rho_k \otimes \rho_\ell) \right].
    \end{align}
    Variances are trivially bounded by the infinity norm squared, immediately giving us the first part of Equation~\eqref{eq:VarianceBoundFull}:
    \begin{align}
        \sum\limits_{k=1}^N \Var(G_k) \leq N \Avg_k \left(||G_k||_\infty^2 \right).
    \end{align}
    The left-hand side is simply the usual, incoherent contribution to a measurement or scattering process, no different from using each component of the target/detector entirely separately; the right-hand side is a (loose) upper bound on it based on the maximum magnitude of the individual interactions.

    We now bound the off-diagonal terms using a series of known inequalities; this approach is inspired in part by Ref.~\cite{Wolf:2007tdq}. Holder's inequality (see e.g. Ref. \cite{SchattenHolderBook}) lets us bound the trace (which is no greater than the Schatten 1-norm) by a product of Schatten norms:
    \begin{align}\begin{split}
        \Tr\left[ G_k^\dagger G_\ell (\rho_{k\ell} - \rho_k \otimes \rho_\ell) \right] &\leq ||G_k^\dagger G_\ell||_\infty ||\rho_{k\ell} - \rho_k \otimes \rho_\ell||_1 \\
        &= ||G_k||_\infty ||G_\ell||_\infty ||\rho_{k\ell} - \rho_k \otimes \rho_\ell||_1.
    \end{split}\end{align}
    The infinity norms here are, again, simply a measure of how large the individual operators are. The more meaningful component is the remaining term, which measures how far the $k$'th and $\ell$'th subsystems are from being in a product state since, for a product of pure states, $\rho_{k\ell} = \rho_k \otimes \rho_\ell$ and thus that term is zero.
    
    The quantum Pinsker inequality (see e.g. Ref. \cite{Watrous}) now bounds the 1-norm by
    \begin{align}
        ||\rho_{k\ell} - \rho_k \otimes \rho_\ell||_1 &\leq \sqrt{2D(\rho_{k\ell} || \rho_k\otimes\rho_\ell)} \label{eq:rhoDiffRelativeEntropy}
    \end{align}
    where
    \begin{align}
        D(\rho || \sigma) = \Tr\left[ \rho\ln\rho - \rho\ln\sigma \right]
    \end{align}
    is the quantum relative entropy between two density matrices. We can rewrite this in our case as
    \begin{align}\begin{split}
        D(\rho_{k\ell} || \rho_k\otimes\rho_\ell) &= \Tr\left[ \rho_{k\ell}\ln\rho_{k\ell} - \rho_{k\ell}\ln(\rho_k\otimes\rho_\ell) \right] \\
        &= \Tr\left[ \rho_{k\ell}\ln\rho_{k\ell} - \rho_{k\ell}\ln(\rho_k\otimes I_\ell) - \rho_{k\ell}\ln(I_k\otimes \rho_\ell) \right] \\
        &= -s_{k\ell} + s_k + s_\ell
    \end{split}\end{align}
    where $s_{k\ell} = -\Tr[\rho_{k\ell}\ln\rho_{k\ell}]$.

    The Araki-Lieb inequality \cite{ArakiLieb, Watrous} now limits $s_{k\ell}$ to the range
    \begin{align}
        |s_k - s_\ell| \leq s_{k\ell} \leq s_k + s_\ell,
    \end{align}
    so we have
    \begin{align}
        D(\rho_{k\ell} || \rho_k\otimes\rho_\ell) &\leq 2\min(s_k, s_\ell).
    \end{align}
    Note that this can become a very loose upper bound if typical subsystems are entangled with only a small number ($\ll N$) other subsystems, such that an order-one fraction of the $D(\rho_{k\ell} || \rho_k\otimes\rho_\ell)$ are zero rather than $\sim\min(s_k,s_\ell)$; this might occur if, for example, entanglement in the system is local. In such systems one could therefore instead rewrite our bound in terms of sums of quantum relative entropies, effectively stopping this derivation at Equation~\eqref{eq:rhoDiffRelativeEntropy}, giving a much stronger bound. We do not emphasize this form of the result, however, since it is often not as convenient or as easily interpretable. Nonetheless this is an important qualitative requirement for achieving large coherent enhancements: one needs not only sufficient entanglement per particle, but that entanglement must be distributed sufficiently broadly.
    
    Combining our results thus far, we have
    \begin{align}
        \sum_{k\neq\ell} \Tr\left[ G_k^\dagger G_\ell (\rho_{k\ell} - \rho_k \otimes \rho_\ell) \right] &\leq 2 \sum_{k\neq\ell} ||G_k||_\infty ||G_\ell||_\infty \sqrt{\min(s_k, s_\ell)}.
    \end{align}
    This expression contains all of the key insights of this result, but it is somewhat unwieldy. We therefore now simply this expression somewhat, at the expense of loosening the inequality, especially for systems in which $||G_k||_\infty$ or $s_k$ vary strongly between subsystems.

    First, note that $\sqrt{\min(s_k, s_\ell)} \leq \sqrt[4]{s_k s_\ell}$. Then, using the Cauchy-Schwartz inequality, we have
    \begin{align}\begin{split}
        \sum_{k\neq\ell} \Tr\left[ G_k^\dagger G_\ell (\rho_{k\ell} - \rho_k \otimes \rho_\ell) \right] &\leq 2\sum_{k \neq \ell} ||G_k||_\infty ||G_\ell||_\infty \sqrt[4]{s_k s_\ell} \\
        &\leq 2 \sqrt{ \left(\sum_{k,\ell} ||G_k||_\infty^2 ||G_\ell||_\infty^2 \right) \left( \sum_{k,\ell} \sqrt{s_k s_\ell} \right) } \\
        &= 2N^2 \Avg\limits_k \left( ||G_k||_\infty^2 \right) \Avg\limits_k\left(\sqrt{s_k}\right).
    \end{split}\end{align}

    At this point we again choose to weaken the result (for significantly varying $s_k$) in order to write it in terms of a somewhat more convenient quantity. Concavity of the square root gives
    \begin{align}
        \Avg\limits_k\left(\sqrt{s_k}\right) \leq \sqrt{ \Avg\limits_k\left(s_k\right) }.
    \end{align}
    We thus have
    \begin{align}
        {\rm Var}(G) = \Tr[\rho G^\dagger G] - \left| \Tr[\rho G] \right|^2 \leq \Avg_k \left(||G_k||_\infty^2 \right) \left( N + 2 N^2 \sqrt{ \Avg\limits_k\left(s_k\right) } \right).
    \end{align}
\end{proof}

As we noted in the course of this proof, we have made several simplifications to our results at the cost of weakening it, especially for systems with strong $k$-dependence. Since the stronger intermediate results may be preferable in some cases, we summarize them here:
\begin{subequations}\begin{align}
    \Var(G) - \sum\limits_{k=1}^N \Var(G_k) &= \sum\limits_{k \neq \ell} \Tr\left[ G_k^\dagger G_\ell (\rho_{k\ell} - \rho_k \otimes \rho_\ell) \right] \\
    &\leq \sum\limits_{k \neq \ell} ||G_k||_\infty ||G_\ell||_\infty ||\rho_{k\ell} - \rho_k \otimes \rho_\ell||_1 \\
    &\leq \sum\limits_{k \neq \ell} ||G_k||_\infty ||G_\ell||_\infty \sqrt{2D(\rho_{k\ell} || \rho_k\otimes\rho_\ell)} \\
    &\leq 2\sum\limits_{k \neq \ell} ||G_k||_\infty ||G_\ell||_\infty \sqrt{\min(s_k,s_\ell)} \\
    &\leq 2N^2 \Avg\limits_k \left( ||G_k||_\infty^2 \right) \Avg\limits_k\left(\sqrt{s_k}\right) \\
    &\leq 2N^2 \Avg\limits_k \left( ||G_k||_\infty^2 \right) \sqrt{\Avg\limits_k\left(s_k\right)}.
\end{align}\end{subequations}
These expressions can easily be substituted into Theorem~\ref{thm:CramerRaoEntropyDensityFull} and Theorem~\ref{thm:MultipleParticleLimitFull} to strengthen those results as well.

Lemma~\ref{lem:OrthogonalScatterProbTotFull} will require a slight generalization of this result, to collections of operators carrying an additional index (the outgoing particle state $\ket{f}$, in scattering contexts):

\begin{manuallemma}{1b}
    For any $\mathcal{H}$ form in Eq.~\eqref{eq:Hdefinition} and let $\{G_k^a\}$ be a finite, $a$-indexed collection of operators on $\mathcal{H}_k$, with 
    \begin{align}
        G^a = \sum\limits_{k=1}^N G_k^a \bigotimes\limits_{\ell \neq k} I_\ell.
    \end{align}
    Then for any density matrix $\rho$ on $\mathcal{H}$,
    \begin{align}
    \begin{split}
        \sum\limits_a {\rm Var}(G^a) \leq \Avg_k \left( D_k^2 \right) \left( N + 2 N^2 \sqrt{ \Avg\limits_k\left(s_k\right) } \right). \label{eq:VarianceBoundVector}
    \end{split}
    \end{align}
    where
    \begin{align}
        D_k^2 = \nnorm{ \sum\limits_a \left(G_k^a\right)^\dagger G_k^a }_\infty.
    \end{align}
    \label{lem:VarianceBoundVector}
\end{manuallemma}

\begin{remark}
    Applying Lemma~\ref{lem:VarianceBoundFull} to each $G^a$ and then summing would give almost the same result, but with
    \begin{align}
        \sum\limits_a \nnorm{ G_k^a }_\infty^2 = \sum\limits_a \nnorm{ \left(G_k^a\right)^\dagger G_k^a }_\infty
    \end{align}
    in place of $D_k^2$, i.e. with the order of the sum and the norm reversed.
\end{remark}

\begin{proof}
    The proof is almost identical to that of Lemma~\ref{lem:VarianceBoundFull}, so we merely summarize the changes. In this case we have
    \begin{align}
        \sum\limits_a \Var(G^a ) &= \sum\limits_a \sum\limits_{k=1}^N \Var(G_k) + \sum\limits_{k \neq \ell} \Tr\left[ \sum\limits_a \left(G_k^a\right)^\dagger G_\ell^a (\rho_{k\ell} - \rho_k \otimes \rho_\ell) \right].
    \end{align}

    For the incoherent first term, note that
    \begin{align}
        \sum\limits_a \Var(G_k) &\leq \sum\limits_a \Tr\left[ \rho_k \left(G_k^a\right)^\dagger G_k^a \right] \\
        &= \Tr\left[ \rho_k \sum\limits_a \left(G_k^a\right)^\dagger G_k^a \right] \\
        &\leq \nnorm{\sum\limits_a \left(G_k^a\right)^\dagger G_k^a}_\infty = D_k^2.
    \end{align}
    Thus
    \begin{align}
        \sum\limits_a \sum\limits_{k=1}^N \Var(G_k) \leq N \Avg_k \left( D_k^2 \right).
    \end{align}

    For the coherent second term, we apply Holder's inequality as before,
    \begin{align}
        \Tr\left[ \sum\limits_a \left(G_k^a\right)^\dagger G_\ell^a (\rho_{k\ell} - \rho_k \otimes \rho_\ell) \right] &\leq \nnorm{ \sum\limits_a \left(G_k^a\right)^\dagger G_\ell^a }_\infty ||\rho_{k\ell} - \rho_k \otimes \rho_\ell||_1,
    \end{align}
    but now apply the row/column operator Cauchy-Schwartz inequality (see e.g. Ref.~\cite{AlgebraBook}) to split the first factor on the right-hand side differently:
    \begin{align}
        \nnorm{ \sum\limits_a \left(G_k^a\right)^\dagger G_\ell^a }_\infty \leq D_k D_\ell.
    \end{align}
    From here the proof proceeds exactly as for Lemma~\ref{lem:VarianceBoundFull}.
\end{proof}

%

\subsection{Lemma \ref{lem:QFIBound}} \label{subapp:QFIProof}

\begin{manuallemma}{2a}
    The maximum QFI of a state $\rho$ is
    \begin{align}
        \frac{ F_Q[\rho, H] }{\Avg\limits_{k}\left( F_Q^\MAX[H_k] \right)} \leq N + 2N^2 \sqrt{\EMF} \label{eq:QFIBoundFull}
    \end{align}
    where $F_Q^\MAX[H_k]$ is the maximum QFI that could be achieved with any state on $\mathcal{H}_k$ alone.
    \label{lem:QFIBoundFull}
\end{manuallemma}

\begin{proof}
    Let $\lambda_{\rm max}^{(k)}$ ($\lambda_{\rm min}^{(k)}$) be the maximum (minimum) eigenvalue of each $H_k$, and define
    \begin{align}
        H_k' = H_k - \pfrac{\lambda_{\rm max}^{(k)} + \lambda_{\rm min}^{(k)}}{2} I_k,
    \end{align}
    such that the largest and smallest eigenvalues of each $H_k'$ differ only by a sign. Now let
    \begin{align}
        H' = \sum\limits_{k=1}^N H_k' \bigotimes\limits_{\ell \neq k} I_\ell.
    \end{align}
    Note that each $H_k'$ differs from $H_k$ only by a multiple of the identity $I_k$, and thus thus is true of the difference between $H'$ and $H$ as well.

    Now let $\ket{\lambda_{\rm max}^{\prime (k)}}$ ($\ket{\lambda_{\rm min}^{\prime (k)}}$) be an eigenstate of $H_k$ with the maximum (minimum) eigenvalue, and let
    \begin{align}
        \ket{\psi_k} = \frac{\ket{\lambda_{\rm min}^{\prime (k)}} + \ket{\lambda_{\rm max}^{\prime (k)}}}{\sqrt{2}}.
    \end{align}
    This is a pure state, and it has variance $\left(\lambda_{\rm max}^{\prime (k)}\right)^2 = \left(\lambda_{\rm min}^{\prime (k)}\right)^2 = ||H_k'||_\infty^2$. We thus have $F_Q[\ket{\psi_k}, H_k'] = 4||H_k'||_\infty^2$, and moreover $F_Q^\MAX[H_k'] = 4||H_k'||_\infty^2$ since the QFI can never exceed four times the variance.

    By Lemma~\ref{lem:VarianceBoundFull}, we have
    \begin{align}
        F_Q[\rho, H'] \leq 4\Avg_k \left(||H'_k||_\infty^2 \right) \left( N + 2 N^2 \sqrt{ \Avg\limits_k\left(s_k\right) } \right) = \Avg_k \left( F_Q^\MAX[H_k'] \right) \left( N + 2 N^2 \sqrt{ \Avg\limits_k\left(s_k\right) } \right).
    \end{align}
    But, since the QFI is unchanged by adding multiples of the identity to an operator, this immediately implies the same relation for the original $H_k$ and $H$:
    \begin{align}
        F_Q[\rho, H] \leq \Avg_k \left( F_Q^\MAX[H_k] \right) \left( N + 2 N^2 \sqrt{ \Avg\limits_k\left(s_k\right) } \right). \label{eq:QFIBoundFullPure}
    \end{align}

    This completes the proof for pure states, but we can make a stronger statement for mixed states using our entanglement of formation, $\EMF$. Consider any decomposition of a mixed state $\rho$ into pure states,
    \begin{align}
        \rho = \sum\limits_\ell p_\ell \ket{\psi_\ell}\bra{\psi_\ell}.
    \end{align}
    Since the QFI is convex under state mixtures \cite{Yu:2013mlg}, Equation~\eqref{eq:QFIBoundFullPure} gives us
    \begin{align}
        F_Q[\rho, H] \leq \Avg_k \left( F_Q^\MAX[H_k] \right) \left( N + 2 N^2 \sum_\ell p_\ell \sqrt{ \Avg\limits_k\left(s_{k,\ell} \right) } \right),
    \end{align}
    where $s_{k,\ell}$ is $s_k$ for the state $\ket{\psi_\ell}\bra{\psi_\ell}$. Another application of Cauchy-Schwarz gives us
    \begin{align}
        F_Q[\rho, H] \leq \Avg_k \left( F_Q^\MAX[H_k] \right) \left( N + 2 N^2 \sqrt{ \sum_\ell p_\ell \Avg\limits_k\left(s_{k,\ell} \right) } \right).
    \end{align}
    Then, since this must hold for all decompositions of $\rho$,
    \begin{align}
        \frac{ F_Q[\rho, H] }{\Avg\limits_{k}\left( F_Q^\MAX[H_k] \right)} \leq N + 2N^2 \sqrt{\EMF}.
    \end{align}
\end{proof}

%

\subsection{Theorem \ref{thm:CramerRaoEntropyDensity}} \label{subapp:CRProof}

\begin{manualtheorem}{3a}
    Given the unitary evolution of a state $\rho$ to $e^{-i\theta H}\rho e^{i\theta H}$, the best precision with which one can determine $\theta$ in $m$ measurements is
    \begin{align}
        (\Delta \theta)^2 \geq \frac{1}{m \Avg\limits_{k}\left( F_Q^\MAX[H_k] \right) \left(N + 2N^2 \sqrt{\EMF} \right)}.
    \end{align}
    \label{thm:CramerRaoEntropyDensityFull}
\end{manualtheorem}

\begin{proof}
    This is simply a substitution of Equation~\eqref{eq:QFIBoundFull} into the quantum Cram{\'e}r-Rao bound (see e.g. Ref.~\cite{QFIReview}).
\end{proof}

%

\subsection{Lemma \ref{lem:OrthogonalScatterProbTot}} \label{subapp:ScatteringProof}

Before we provide a generalized form for and proof of Lemma~\ref{lem:OrthogonalScatterProbTot}, we discuss some technical requirements for the single-particle cross-sections required in order for our results to hold. In the main text, we considered a scattering process with S-matrix
\begin{align}
    S = 1 + i\alpha \sum\limits_{k=1}^N T_k^{(1)} + O(\alpha^2),
\end{align}
although we restricted to the case of all $T_k$ equal, which we will not do here. We then defined an orthogonal scattering probability, glossing over the distinction between that probability for a single outgoing particle state $\ket{f}$, and the probability summed or integrated over those states. We now handle this more precisely.

For a particular outgoing particle state $\ket{f}$, we will write the orthogonal scattering probability as
\begin{align}
    p_\perp^{i \to f} = \bra{i}\bra{t} S^\dagger \left( \ket{f}\bra{f}\otimes \Pi_\perp \right) S \ket{i}\ket{t}, \label{eq:pPerpif}
\end{align}
omitting the superscipt when summing\footnote{Our results formally hold only for finite-dimensional spaces of possible outgoing particle states. Although real physical systems often have continuous spaces of possible final states, in practice it should essentially always be possible to remedy this with appropriate spatial and energy cutoffs, so long as any scattering poles are regulated. In the continuous final state case, the single-final state probability should then be understood as an infinitesimal carrying the appropriate delta functions for momentum and energy conservation, but these are uninteresting for our purposes.} over $\ket{f}$, i.e. $p_\perp$.

We assume that $p_\perp \ll 1$, noting that our scaling bounds became trivially true if that is not the case, since ultimately $p_\perp \leq 1$ regardless of $N$. This limit will always hold for sufficiently small $\alpha$, which is the regime where our results are sensible.

For any particular incident and outgoing states $\ket{i}$ and $\ket{f}$, we additionally define the maximum orthogonal scattering probability (or probability density) of each individual target particle, over all possible states of that particle, to be $p_{\perp,k}^{\rm max}(i\to f)$. We use $p_{\perp,k}^{\rm max}$ to refer to the sum or integral of $p_{\perp,k}^{\rm max}(i\to f)$ over final states $\ket{f}$, and require $p_{\perp,k}^{\rm max} \ll 1$ as above. 

We additionally assume that $p_{\perp,k}^{\rm max}$ is $O(1)$ with respect to $N$, i.e. that there is some $N$-independent upper bound on it that holds for all $k$. Although this last assumption is not necessary for Lemma~\ref{lem:OrthogonalScatterProbTotFull}, its absence would render the lemma misleading by allowing $N$-dependence to hide within the factor of $\Avg_k( p_{\perp,k}^{\rm max} )$. Moreover it \textit{is} necessary for the subsequent results.

\begin{manuallemma}{4a}
    For any scattering process of the form above, the \textit{total} orthogonal scattering probability from a pure target state $\rho$ at order $\alpha^2$ is upper bounded by
    \begin{align}
        p_\perp \leq \Avg\limits_k\left( p_{\perp,k}^{\rm max} \right) \left( N + 2 N^2 \sqrt{ \Avg\limits_k\left(s_k\right) } \right). \label{eq:OrthogonalScatterProbTotFull}
    \end{align}
    \label{lem:OrthogonalScatterProbTotFull}
\end{manuallemma}

\begin{proof}
    Let
    \begin{align}
        T_k^{i \to f} \equiv \bra{f} T_k^{(1)} \ket{i}
    \end{align}
    be the restriction of $T_k^{(1)}$ to the target space for particular incident and outgoing states, and let
    \begin{align}
        T^{i\to f} = \sum\limits_{k=1}^N T_k^{i \to f} \bigotimes\limits_{\ell \neq k} I_\ell.
    \end{align}
    Then, rewriting Eq.~\eqref{eq:pPerpSimplified}, we have
    \begin{align}
        p_\perp = \alpha^2 \sum\limits_f \Var\left(T^{i\to f}\right) + O(\alpha^3).
    \end{align}
    For brevity, we omit the $O(\alpha^3)$ term from here on.

    Lemma~\ref{lem:VarianceBoundVector} gives us
    \begin{align}
        \sum\limits_f \Var\left(T^{i\to f}\right) \leq \Avg\limits_k\left( \nnorm{\sum\limits_f \left(T_k^{i \to f}\right)^\dagger T_k^{i \to f} }_\infty \right) \left( N + 2 N^2 \sqrt{ \Avg\limits_k\left(s_k\right) } \right).
    \end{align}
    Moreover, since the variance of $T$ is not affected by changing the $T_k$ by multiples of the identity, we must have
    \begin{align}
        \sum\limits_f \Var\left(T^{i\to f}\right) \leq \Avg\limits_k\left( \nnorm{\sum\limits_f \left(T_k^{i \to f} - c_k^f I_k \right)^\dagger \left(T_k^{i \to f} - c_k^f I_k\right) }_\infty \right) \left( N + 2 N^2 \sqrt{ \Avg\limits_k\left(s_k\right) } \right).
    \end{align}
    for any set of complex values $\{c_k^f\}$ indexed by $k$ and $f$. Conversely, the maximum orthogonal scattering probability achievable from any state of $\mathcal{H}_k$ alone is
    \begin{align}
        p_{\perp, k}^{\rm max} &= \alpha^2 \max\limits_{\rho_k} \left( \sum\limits_f \Var\left(T_k^{i\to f}\right) \right).
    \end{align}
    These can be related through Lemma~\ref{lem:MinimaxLemma}, proven in the last subsection of this appendix, which gives
    \begin{align}
        \max\limits_{\rho_k} \left( \sum\limits_f \Var\left(T_k^{i\to f}\right) \right) = \min_{\{c_k^f \in \mathbb{C}\}} \nnorm{ \sum\limits_f \left(T_k^{i\to f}-c_k^f I_k\right)^\dagger\left(T_k^{i\to f}-c_k^f I_k\right) }_\infty.
    \end{align}
    We thus have
    \begin{align}
        p_\perp \leq \Avg\limits_k\left( p_{\perp,k}^{\rm max} \right) \left( N + 2 N^2 \sqrt{ \Avg\limits_k\left(s_k\right) } \right).
    \end{align}
\end{proof}

%

\subsection{Lemma \ref{lem:TraceDistanceBound}} \label{subapp:TraceDistanceBoundProof}

\begin{manuallemma}{5a}
    For any scattering process of the same form as in Lemma~\ref{lem:OrthogonalScatterProbTotFull}, the trace distance between the final states in the presence ($\rho_f'$) and absence ($\rho_f$) of the incident particle is upper bounded by
    \begin{align}
        \TrD{\rho_f'}{\rho_f} \leq \sqrt{p_\perp^{i\to i}} + \sqrt{p_\perp p_s} \label{eq:TraceDistanceBoundAppendix}
    \end{align}
    where $p_\perp^{i\to i}$ is the part of $p_\perp$ coming from $\ket{f} = \ket{i}$ (i.e. fully forward elastic scattering of the incident particle) and $p_s$ is the total probability of scattering to an outgoing particle state different from the incident particle state.
    \label{lem:TraceDistanceBoundAppendix}
\end{manuallemma}

\begin{proof}
    We work in the interaction picture, so $\rho_f = \ket{t}\bra{t}.$ Now let $\{\ket{g}\}$ be an orthonormal basis of outgoing states containing $\ket{i}$, and let
    \begin{align}
        S_g = \bra{g}S\ket{i}
    \end{align}
    be a set of Kraus operators \cite{Watrous}. Note that this is a decomposition over final states of the outgoing particle, rather than a decomposition over subsystems, as we do in most of this work. Then we have
    \begin{align}
        \rho_f' = \sum_g S_g \rho S_g^\dagger. \label{eq:RhoFPSgSum}
    \end{align}
    We decompose this into $\rho_f' = \sigma_0 + \sigma_1$ where
    \begin{subequations}
        \begin{align}
            \sigma_0 &= S_i \rho S_i^\dagger \\
            \sigma_1 &= \sum_{g \neq i} S_g \rho S_g^\dagger.
        \end{align}
    \end{subequations}
    Note that $\sigma_0$ and $\sigma_1$ are not density matrices, as they have trace less than one; in particular, $\Tr[\sigma_1] = p_s$. We therefore also define their normalized forms,
    \begin{subequations}
        \begin{align}
            \rho_0 &= \frac{\sigma_0}{1-p_s} \\
            \rho_1 &= \frac{\sigma_1}{p_s}.
        \end{align}
    \end{subequations}

    We likewise divide the orthogonal scattering probability, $p_\perp = \Tr[\rho'(1-\rho)] = 1 - \bra{t}\rho_f'\ket{t}$ into
    \begin{subequations}
        \begin{align}
            p_\perp^{(0)} =  p_\perp^{i\to i} &= \Tr[\sigma_0(1-\rho)] \\
            p_\perp^{(1)} &= \Tr[\sigma_1(1-\rho)].
        \end{align}
    \end{subequations}

    Then, since the trace distance is convex under state mixture, 
    \begin{align}\begin{split}
        \TrD{\rho_f'}{\rho_f} &\leq (1-p_s)\TrD{\rho_0}{\rho_f} + p_s\TrD{\rho_1}{\rho_f} \\
        &\leq (1-p_s)\sqrt{\frac{p_\perp^{(0)}}{1-p_s}} + p_s\sqrt{\frac{p_\perp^{(1)}}{p_s}} \\
        &\leq \sqrt{p_\perp^{(0)}} + \sqrt{p_\perp p_s} \label{eq:TraceDistanceBoundAppendixEnd}
    \end{split}\end{align}
    where in the second line we have taken advantage of the purity of $\rho_f$ and the Fuchs-van de Graaf inequalities \cite{FuchsVDG},
    \begin{align}
        \TrD{\rho_1}{\rho_f} &\leq \sqrt{1 - \bra{t}\rho_1\ket{t}}.
    \end{align}
    (In fact, Equation~\eqref{eq:TraceDistanceBoundAppendixEnd} can be marginally strengthened by replacing $p_\perp$ with $p_\perp^{(1)}$, but we will not use this fact.)
\end{proof}

%

\subsection{Lemma \ref{lem:SingleParticleGeneralLimit}} \label{subapp:SingleParticleGeneralLimitProof}

Before presenting a proof of the generalized Lemma~\ref{lem:SingleParticleGeneralLimit}, we briefly discuss the meaning of the power $\pi$ and type I error $\epsilon$ which are set to some non-zero constant within it. Consider an experiment in which a target is observed for some period of time before its state is read out. We assume in this section that at most one particle could have been incident during this period, such that the two possible final states are $\rho_f$ (if there was not an incident particle) and $\rho_f'$ (if there was one). The target is measured in some way, giving either a positive or negative result, based on which the observer attempts to determine whether the target's final state was $\rho_f$ or $\rho_f'$. The power $\pi$ is the probability that, conditional on there being an incident particle, the measurement (correctly) gives a positive result. Conversely, the type I error $\epsilon$ is the probability, conditional on there \textit{not} being an incident particle, that the the measurement nonetheless gives a positive result. Below we show that our methods allow us to bound the quantity $\pi - \epsilon$. Note that this quantity must be positive for any meaningful measurement, since one can achieve any $\pi = \epsilon$ simply by choosing the outcome randomly (positive with probability $\pi = \epsilon$), independent of the observation.

From a frequentist perspective, the type I error is simply the p-value threshold. A typical notion of ``expected sensitivity'' is the edge of where one can see 50\% of signals at the given confidence level, corresponding to $\pi = 1/2$ and $\epsilon \ll 1$. One therefore requires $\pi - \epsilon \sim 0.5$ at the projected sensitivity limit.

Translating into a Bayesian update requires one additional input: the prior $q$ on there being an incident particle. From there, however, it is straightforward to calculate posteriors, conditional on the observer measuring that the final state was $\rho_f$ or $\rho_f'$:
\begin{subequations}
\begin{align}
    P({\rm signal~present}|{\rm measure~}\rho_f) &= \frac{q(1-\pi)}{q(1-\pi) + (1-q)(1-\epsilon)} \\
    P({\rm signal~present}|{\rm measure~}\rho_f') &=  \frac{q\pi}{q\pi + (1-q)\epsilon}.
\end{align}
\end{subequations}
Both posteriors go to $q$, i.e. the update becomes increasingly weak, as $\pi - \epsilon \to 0$.

In the proof below (as well as the subsequent proof of Theorem~\ref{thm:MultipleParticleLimitFull}) we will simply require $\pi-\epsilon$ to be some chosen positive constant, remaining agnostic to its specific value.

\begin{manuallemma}{6a}
    For any chosen type I error rate and power, no experiment searching for a single particle scattering at $O(\alpha)$ from an $N$-particle detector can have a detection limit better than
    \begin{align}
        \frac{1}{\alpha_{\rm min}^2} = O\left( N + N^2 \sqrt{\EMF} \right).
    \end{align}
    Moreover, if a scattering process has $p_\perp^{i\to i} = 0$ for all pure state components of the initial state (i.e. if fully forward elastic scattering has no effect on the target), then no experiment can detect an effect below $O(\alpha^2)$ and, at that order, the best possible scaling of the detection limit is
    \begin{align}
        \frac{1}{\alpha_{\rm min}^2} = O\left( N^{3/2} + N^2 \EMF^{1/4} \right).
    \end{align}
    \label{lem:SingleParticleGeneralLimitFull}
\end{manuallemma}

\begin{proof}
    Suppose we perform some experiment that attempts to determine whether or not a particle was incident on a target. To do so, we must be able to distinguish the final state $\rho_f'$ of the target which is attained in the presence of such an incident particle from the final state $\rho_f$ attained if there is no such incident particle. Note that this is distinct from the question of whether or not a scattering process actually occurred which, as we discuss in the main text, may be basis-dependent. Instead, we absorb the notion of ``scattering probability'' into the \textit{mixed} final state $\rho_f'$. The question of what values of $\alpha$ can be detected is thus equivalent to the question of what values of $\alpha$ are sufficient for $\rho_f'$ to be distinguished from $\rho_f$.

    Lemmas~\ref{lem:OrthogonalScatterProbTotFull} and~\ref{lem:TraceDistanceBoundAppendix} together give us a bound on the trace distance between $\rho_f$ and $\rho_f'$ (assuming for now that the initial state is pure; we return to mixed states below). Translating this into a bound on detection limits requires one to choose error rates, since these affect the resulting limits. If we require a power $\pi$ and a type I error rate $\epsilon$, a standard generalization of Helstrom's Theorem (see e.g. Ref. \cite{Wilde}) tells us that we must have
    \begin{align}
        \frac{1}{2}||\rho_f - \rho_f'||_1 \geq \pi - \epsilon,
    \end{align}
    for any possible measurement strategy (i.e. any positive operator-valued measure). This gives a lower bound on $d_B(\rho_f, \rho_f')$, and thus a lower bound on $\alpha$ for a given scaling of the scattering probability with $\alpha$, $N$, and $\EMF$.

    We begin by considering pure initial states, such that our results about orthogonal scattering probabilities are applicable. For experiments operating at $O(\alpha)$, only the $\sqrt{p_\perp^{(0)}}$ component of Eq.~\eqref{eq:TraceDistanceBoundAppendix} contributes. Then since $p_\perp^{(0)} \leq p_\perp = O(\alpha^2 N + \alpha^2 N^2 \sqrt{ \Avg_k(s_k) })$, we have
    \begin{align}
        \pi - \epsilon \leq C \sqrt{\alpha^2 N + \alpha^2 N^2 \sqrt{ \Avg_k(s_k) })}
    \end{align}
    for some constant $C$, and thus
    \begin{align}
        \frac{1}{\alpha_{\rm min}^2} = O\left( N + N^2 \sqrt{ \Avg\limits_k\left(s_k\right) } \right).
    \end{align}

    Experiments where $p_\perp^{(0)} = 0$ must instead rely on the $\sqrt{p_\perp p_s}$ term (or on terms at $O(\alpha^3)$ or above). Then, since $p_s \lesssim \alpha^2 N^2$, we have $p_\perp p_s = O(\alpha^4 N^3 + \alpha^4 N^4 \sqrt{ \Avg_k(s_k) })$ and thus
    \begin{align}
        \frac{1}{\alpha_{\rm min}^2} = O\left( N^{3/2} + N^2 \left( \Avg\limits_k\left(s_k\right) \right)^{1/4} \right).
    \end{align}

    Up to this point we have restricted ourselves to pure initial states. Now recall that the trace distance is linear under scalar multiplication and satisfies the triangle inequality; thus the trace distance between a mixture of pure states and its corresponding final state is upper bounded by the initial probability-weighted average of the trace distances for each initial pure state component. We can thus replace $\Avg_k(s_k)$ with $\EMF$ in the inequalities above, just as we did in the proof of Lemma~\ref{lem:QFIBoundFull}, giving our desired results.
\end{proof}

%

\subsection{Theorem \ref{thm:MultipleParticleLimit}} \label{subapp:MultipleParticleLimitProof}

\begin{manualtheorem}{7a}   
    For any chosen type I error rate and power, no experiment searching for $K$ independent incident particles scattering at $O(\alpha)$ from an $N$-particle detector can have a detection limit better than
    \begin{align}
        \frac{1}{\alpha_{\rm min}^2} = O\left( K^2 N + K^2 N^2 \sqrt{\EMF^\MAX} \right),
    \end{align}
    with $\EMF^\MAX$ the largest value of $\EMF$ achieved over the integration time. Moreover, if a scattering process has $p_\perp^{i\to i} = 0$ for all pure state components of the initial state (i.e. if fully forward elastic scattering has no effect on the target), then no experiment can detect an effect below $O(\alpha^2)$ and, at that order, the best possible scaling of the detection limit is
    \begin{align}
        \frac{1}{\alpha_{\rm min}^2} = O\left( K N^{3/2} + K N^2 \left(\EMF^\MAX\right)^{1/4} \right).
    \end{align}
    \label{thm:MultipleParticleLimitFull}
\end{manualtheorem}

\begin{proof}
    We prove the $O(\alpha)$ case; the proof for $O(\alpha^2)$ is essentially identical. Let $\rho_f^{(K)}$ be the state after $K$ incident particles, with $\rho_f^{(0)} = \rho_f$. Then the triangle inequality on trace distances \cite{Watrous} ensures that
    \begin{align}
        \frac{1}{2}||\rho_f^{(K)} - \rho_f||_1 &\leq \sum\limits_{\ell=1}^K \frac{1}{2}||\rho_f^{(\ell)} - \rho_f^{(\ell-1)}||_1.
    \end{align}
    We showed in the proof of Lemma~\ref{lem:SingleParticleGeneralLimitFull} that each of these terms is bounded by
    \begin{align}
        \frac{1}{2}||\rho_f^{(\ell)} - \rho_f^{(\ell-1)}||_1 &\leq O\left( \alpha\sqrt{N} + \alpha N \sqrt[4]{\EMF^{(\ell-1)}} \right),
    \end{align}
    where $\EMF^{(\ell-1)}$ is the value of $\EMF$ after $\ell-1$ incident particles. For a simple bound, we simply take the maximum achieved value of $\EMF$, in which case we have
    \begin{align}
        \frac{1}{2}||\rho_f^{(K)} - \rho_f||_1 &= O\left( K\alpha\sqrt{N} + K\alpha N \sqrt[4]{\EMF^\MAX} \right).
    \end{align}
    From here the proof is exactly as in Lemma~\ref{lem:SingleParticleGeneralLimitFull}.
\end{proof}

%

\subsection{Lemma 8} \label{subapp:MinimaxLemma}

The proof of Lemma~\ref{lem:OrthogonalScatterProbTotFull} requires a minor generalization of Theorem~9 of Ref.~\cite{AUDENAERT20101126}:

\begin{manuallemma}{8}
    Let $\{G^a\}$ be a finite, $a$-indexed collection of operators on a finite-dimensional Hilbert space $\mathcal{H}$. Then
    \begin{align}
        \min_{\{c^a \in \mathbb{C}\}} \nnorm{ \sum\limits_a (G^a-c^a I)^\dagger(G^a-c^a I) }_\infty = \max\limits_\rho \sum_a \Var_\rho(G^a)
    \end{align}
    where the minimum is over collections of complex numbers, also indexed by $a$, while the maximum is over the state $\rho$ of $\mathcal{H}$ used in the variance.
    \label{lem:MinimaxLemma}
\end{manuallemma}

\begin{proof}
    The proof is closely analogous to that in Ref.~\cite{AUDENAERT20101126}. Let
    \begin{align}
        F(\rho,\{c^a\}) = \sum_a \Tr\left[ \rho (G^a-c^a I)^\dagger(G^a-c^a I) \right].
    \end{align}
    For fixed $\{c^a\}$, we have
    \begin{align}
        \max\limits_\rho F(\rho,\{c^a\}) = \nnorm{ \sum_a (G^a-c^a I)^\dagger(G^a-c^a I) }_\infty
    \end{align}
    while, for fixed $\rho$, we have
    \begin{align}
        \min_{\{c^a\}} F(\rho,\{c^a\}) = \sum\limits_a \Var_\rho(G^a),
    \end{align}
    obtained when $c^a = \Tr[\rho G^a]$.

    $F$ is linear in $\rho$ and convex in $\{c^a\}$, and the space of density matrices is compact, but to apply the minimax theorem (see e.g. Ref.~\cite{Sion1958}), we need one additional requirement: a restriction of the $\{c^a\}$ to a compact set. Let
    \begin{align}
        M \equiv \nnorm{ \sum\limits_a \left(G^a\right)^\dagger G^a }_\infty.
    \end{align}
    Some straightforward algebra shows that
    \begin{align}
        \max\limits_\rho F(\rho,0) = M
    \end{align}
    while, for $\nnorm{c}^2 = \sum_a |c^a|^2 > 4M$,
    \begin{align}
        \max\limits_\rho F(\rho,c^a) > M,
    \end{align}
    so the minimum of the maximum always occurs for $\nnorm{c} \leq 2\sqrt{M}$. Likewise, the fixed-$\rho$ minima found above occur at $\nnorm{c}^2 = \sum_a | \Tr[\rho G^a] |^2 \leq M$. We can thus choose a ball of radius $2\sqrt{M}$ for the $\{c^a\}$, providing our compact set.

    We can therefore apply the minimax theorem to $F$, giving
    \begin{align}
        \max\limits_\rho \min_{\{c^a\}} F(\rho,\{c^a\}) = \min_{\{c^a\}} \max\limits_\rho F(\rho,\{c^a\})
    \end{align}
    and thus
    \begin{align}
        \min_{\{c^a \in \mathbb{C}\}} \nnorm{ \sum\limits_a (G^a-c^a I)^\dagger(G^a-c^a I) }_\infty = \max\limits_\rho \sum_a \Var_\rho(G^a).
    \end{align}
\end{proof}

\section{Indistinguishable Bosons}\label{app:IndistinguishableBosons}

Although our main text results can in principle be applied to systems of indistinguishable bosons, their interpretation in this case can become rather subtle, as the meaning of $N$---as defined by the way in which the Hilbert space, interaction, etc. factor; see Equations~\eqref{eq:Hdefinition} and~\eqref{eq:Gdefinition}---may correspond to the number of distinct bosonic modes, as opposed to the number of particles (e.g. atoms) in the detector. We illustrate this with two closely-related examples:

First, consider an atom interferometer using $N$ spin-1 atoms, which begins with them in a twin-Fock state
\begin{align}
    \ket{t} = \ket{N_{m_F=+1} = \frac{N}{2}, N_{m_F=0}=0, N_{m_F=-1}=\frac{N}{2}},
\end{align}
where $N_{m_F = +1}$ is the occupation number of the mode with spin state $m_F = +1$, and likewise for $m_F = 0,-1$. We assume that the atoms and their spatial wavefunctions are identical (i.e. Bose-Einstein condensation), differing only in their spin states, so these three occupation numbers completely specify the state.

As written, this is not an entangled state---it's simply the product state of three Fock states of different modes---and yet it has a QFI with respect to rotations of the spins (e.g. $J_x$) that scales with $N^2$ (see e.g. Ref. \cite{InterferometerCoherenceReview}). This appears to be a contradiction of our Lemma~\ref{lem:QFIBound}, but it is not: the factorization of the Hilbert space is over three modes, not over the $N$ atoms (i.e. this system actually has $N=3$ in the language of the main text, independent of the number of atoms), and moreover the interaction mixes these modes (because it transfers atoms from one mode to another, rather than acting on one mode at a time). Our results are thus not directly applicable to this case.

Interestingly, however, this inapplicability can be remedied using an alternative description of this system: A system of $N$ indistinguishable bosons, each individually in Hilbert space $\mathcal{H}_1$, lives in the symmetric subspace ${\rm Sym}^N(\mathcal{H}_1) \subset \mathcal{H}_1^N$. We can therefore choose to work in the larger space $\mathcal{H}_1^N$, i.e. to treat the atoms as distinguishable, and simply restrict to symmetric states and interactions; this is equivalent to working in ${\rm Sym}^N(\mathcal{H}_1)$ to begin with, and is not an uncommon choice in existing literature (see e.g. Refs. \cite{InterferometerCoherenceReview, Luo:2017nnb, Killoran:2014snn}). Interactions that transfer one atom from one mode to another in the description above (e.g. our rotation $J_x$) can then be written in the factored form of Eq.~\eqref{eq:Gdefinition}, simply by treating that transfer as a change in the state of a particular atom and then summing over possible atoms. From there, our results can be applied without issue, with $N$ now referring to the number of atoms rather than the number of modes; the state above, for example, has $\rho_k = {\rm diag}(1/2,0,1/2)$ and thus $s_k = \ln 2$ in this language, and therefore the $N^2$ scaling of the QFI is entirely consistent with our bounds. The only tradeoff for this is that $s_k$ becomes a somewhat less interpretable object, since it is specific to the embedding and not a property of a single physical particle.

In some contexts this prescription is somewhat more awkward, however. A multimode electromagnetic cavity prepared in the state
\begin{align}
    \ket{t} = \ket{N_1 = \frac{N}{2}, N_2=\frac{N}{2}},
\end{align}
with $N_{1,2}$ the occupation numbers of two different modes, is again in an unentangled state but can have an $O(N^2)$ QFI with respect to a Hamiltonian such as $a^\dagger b + {\rm h.c.}$, with $a^\dagger$ the raising operator on mode 1 and $b$ the lowering operator on mode 2. Although this is entirely analogous to the atomic case above, and can be handled in the same way (i.e. by treating individual photons as distinguishable), in the photon case this becomes rather unnatural. One can therefore choose to either use this atypical description or to accept that interactions coupling distinct photon modes do not satisfy the conditions for our bounds to hold.

This issue does not arise for indistinguishable fermions because each fermion must necessarily be in its own state. One can therefore simply index the fermions by those distinct states and treat them as distinguishable; this generally prevents large coherent enhancements for elastic scattering, as distinct initial states lead to distinct final states (at a given momentum transfer), but it allows for coherent enhancements when internal degrees of freedom (e.g. spin) are present, with the spins treated as distinguishable thanks to the distinct position/momentum states of the particles.

\section{Bound State Entanglement}\label{app:BoundStateEntanglement}

Perhaps the most widespread form of entangled states that can produce coherent enhancements is bound states. Although this is not often recognized, any quantum mechanical bound state is necessarily an entangled state: in an unentangled multiparticle pure state (i.e. a product state), the particles are uncorrelated and consequently independent of one another, precluding any binding interaction between them.

We can describe any bound state of two distinguishable particles (omitting any internal degrees of freedom of the individual particles) in 1D by its momentum-space wavefunction
\begin{equation}
    \ket{\psi} = \int dP \int d\Delta p\, \phi_\tot(P) \phi_{\rm rel}(\Delta p) \ket{\frac{P}{2}+\Delta p}_1 \ket{\frac{P}{2}-\Delta p}_2
\end{equation}
where $P$ is the total momentum, $\Delta p$ is a relative momentum, and the ket subscripts label the two particles. One can immediately see that for most (although not all) choices of $\phi_\tot(P)$ and $\phi_{\rm rel}(\Delta p)$, knowledge of one particle's momentum gives information about the other particle's momentum, and thus the particles are entangled.

As a simple example, consider states where $\phi_\tot(P)$ and $\phi_{\rm rel}(\Delta p)$ are normalized gaussians with variances $\Var(P)$ and $\Var(\Delta p)$, respectively. We illustrate this situation qualitatively as one of our examples in Appendix~\ref{app:IntuitivePictures}. Evaluating the entropy explicitly is possible via Lemma~\ref{lem:GaussianEntanglement} below: in particular, when $\Var(\Delta p) \gg \Var(P)$---as is often the case when, for example, a bound state is at a low temperature compared to its internal energy scales---this gives a per-particle von Neumann entanglement entropy of
\begin{equation}
    s_k \gtrsim 1 + \frac{1}{2}\ln\pfrac{\Var(\Delta p)}{4\Var(P)}.
\end{equation}

Realistic physical systems are of course often not gaussian, but the non-gaussian generalization will be similarly entangled. For any system, being bound necessarily requires entanglement, and this allows bound states to give coherent enhancements in a wide variety of contexts---such as the CEvNS example in the main text.

\begin{manuallemma}{9}
    Consider an $N$-particle gaussian wavefunction (i.e. $\psi(\mathbf{p}) \propto \exp(-\mathbf{p}^T A \mathbf{p}/2)$ for some matrix $A$, with $\mathbf{p}$ the vector of all constituent momenta) with center-of-mass position (momentum) $X$ ($P$), and let $x$ ($p$) be the position (momentum) of any chosen constituent particle (along a chosen axis, if working in more than 1D). Then, if $\Var(p) \gg \Var(P)$, the von Neumann entanglement entropy of the chosen particle is
    \begin{align}
        s \gtrsim 1 + \frac{1}{2}\ln\left(\frac{{\rm Var}(p)}{4{\rm Var}(P)}\right)
    \end{align}
    in one spatial dimension, summed over the spatial dimensions if there are multiple.
    \label{lem:GaussianEntanglement}
\end{manuallemma}

\begin{proof}
    Since our state is gaussian, the reduced state of our single chosen constituent particle, after tracing out the $N-1$ others, is necessarily also gaussian. Crucially, it can carry no covariance between $x$ and $p$: Since we chose our wavefunction to be gaussian with no linear bias, it is real-valued in both position and momentum space. Now note that ${\rm Cov}(x,p) = {\rm Re}(\expect{xp} - \expect{x}\expect{p})$, but $\expect{p}=0$ and $p = -i\partial/\partial x$ so $\expect{xp} = -i\int \psi(x) x \psi'(x)$ is purely imaginary, and therefore ${\rm Cov}(x,p) = 0$. Thus the covariance matrix of our reduced state is simply ${\rm diag}({\rm Var}(x), {\rm Var}(p))$, with symplectic eigenvalue $\nu = \sqrt{{\rm Var}(x){\rm Var}(p)}$. 
    
    It is known that the von Neumann entropy of a gaussian state with symplectic eigenvalues $\nu_k$ is 
    \begin{align}
        s = \sum_k \left(\nu_k + \frac{1}{2}\right) \ln\left(\nu_k + \frac{1}{2}\right) - \left(\nu_k - \frac{1}{2}\right) \ln\left(\nu_k - \frac{1}{2}\right)
    \end{align}
    (see e.g. Refs.~\cite{Weedbrook:2011wxo, Holevo:1999sqc}, although note the varying factor of 2 conventions).

    Now note that we can always write our full state as $\ket{\psi_{\rm CM}} \otimes \ket{\psi_{\rm rel}}$, with $\ket{\psi_{\rm CM}}$ the wavefunction of the center-of-mass and $\ket{\psi_{\rm rel}}$ the wavefunction of the internal degrees of freedom; this factorization must hold since the binding potential is a function of relative coordinates, not of the center-of-mass state. Then the variance of $x$ is equal to the variance of $X$ plus the variance of the relative position $x-X$, and thus ${\rm Var}(x) \geq {\rm Var}(X)$. But the uncertainty principle gives $\sqrt{{\rm Var}(X){\rm Var}(P)} \geq 1/2$. Putting these together, we have
    \begin{align}
        \nu \geq \sqrt{\frac{{\rm Var}(p)}{4{\rm Var}(P)}}
    \end{align}
    for our chosen particle and axis. In the frequently-occurring case ${\rm Var}(p) \gg {\rm Var}(P)$ (as in the nuclear example above) and thus $\nu \gg 1$, we can then approximate
    \begin{align}
        s = \sum_k \ln(\nu_k) + 1 + O\left(\frac{1}{\nu_k^2}\right)
    \end{align}
    so each spatial dimension contributes
    \begin{align}
        s \gtrsim 1 + \frac{1}{2}\ln\left(\frac{{\rm Var}(p)}{4{\rm Var}(P)}\right)
    \end{align}
    to the single-particle von Neumann entanglement entropy. Note in particular that this has no $N$-dependence, and thus the entanglement entropy of each particle is (at least approximately) constant as one increases the system size, in spite of only imposing a single constraint---small total momentum uncertainty---on the growing bound state.
\end{proof}

\section{Intuitive Pictures of the Effects}\label{app:IntuitivePictures}

In this appendix, we provide some intuition for the bounds presented in the main text, which might otherwise appear somewhat opaque. Coherent enhancements in general are familiar: they occur whenever the final states resulting from interactions with each of the $N$ individual target or detector particles have significant overlap, such that one should sum amplitudes rather than probabilities. In this work, however, we focus on processes that are not only coherent, but also detectable; in particular, we require the final state (after the interaction or scattering) to be distinguishable from the initial state (i.e. the state in the absence of an interaction or scattering, since we work in the interaction picture of quantum mechanics) with non-vanishing probability in the large-$N$ limit. In this appendix, we will refer to the requirement of $N^2$ scaling as ``coherence'' and the requirement of $O(1)$ difference between initial and final states as ``detectability'' for brevity, treating both as binary (yes/no) rather than matters of degree for simplicity.

In the rest of this appendix, as well as in the figure captions, we describe the processes in terms of scattering for concreteness. All of the discussion is essentially unchanged for a parameter estimation setup however, with the question of which target particle a scattering contribution arose from replaced by the question of which contribution to the system's evolution (i.e. which $H_k$) is being considered. The only exception to this is the discussion of entanglement accumulation late in this appendix, as that cannot occur for coherent evolution: each $H_k$ affects its corresponding subsystem independently, and thus $H$ cannot produce entanglement between them; scattering differs in this regard because there is an interaction between each subsystem and the (shared) scattered particle, which acts as an ancilla and can therefore lead to changes in $s_k$.

We first note that many of the most familiar coherent processes are not detectable (in the $O(1)$ sense), and likewise many detectable processes are not coherent. We illustrate this distinction in Table~\ref{tab:IntuitionTable}, which considers a process that transfers some fixed momentum $\Delta p$ to either of $N=2$ particles in 1D. (This is the motivating example in the main text.) In the first row, we present a standard incoherent process (e.g. hard scattering): the momentum transfer is much larger than any momentum uncertainty in the initial state, and thus wavefunctions resulting from interacting with the $k=1$ versus $k=2$ particles are both essentially orthogonal to both the initial state and each other. (In the absence of decoherence, the final state will be a superposition of the blue and green states; we refer to them as separate final states, however, to emphasize that they result from different terms in the separable interaction.) The change in target state is thus reliably detectable (at least in principle), but there is no coherent enhancement from constructive interference. 

Conversely, in the second row, we consider a fairly typical coherent process (e.g. very soft scattering). In this case the two possible final states are both nearly identical to the initial state. They thus also have high overlap with one another---resulting in coherence---but a single scatter can no longer be detected reliably. Broadly speaking, one has two options in this case: One can perform measurements sensitive to the full $O(N^2)$ scattering rate---e.g. by measuring the total momentum of the target $\sum_k p_k$, which changes expectation every scatter---but this leads to a trace distance per scatter that falls as $O(N^{-1/2})$ due to the increasing variance of $\sum_k p_k$ with $N$. Alternatively one may sometimes be able to perform measurements based on momenta not present in the initial distribution---e.g. values of $p_k$ just above some cutoff present in the initial state, perhaps due to a finite trapping potential---but such contributions will occur at only an $O(N)$ rate. Thus the former strategy is parametrically superior to both the latter and to the incoherent scattering process of the first row, offering a trace distance increase per incident particle of $O(N^{3/2})$, rather than $O(N)$. Note that, in spite of the large overlap between the initial and final target states in this case, the incident and outgoing states of the scattered particle may be orthogonal, and thus one might reliably detect that this process occurred using a separate detector. What we are considering, however, is the (final) detector itself, in which detectability of the target/detector state change itself is necessary.

As we show in the main text, achieving coherence and $O(1)$ detectability simultaneously requires entanglement. An example of this is presented in the third row of Table~\ref{tab:IntuitionTable}, in which the two particles are entangled such that they have large individual momentum uncertainties but little uncertainty in their total momentum. The possible final states thus exhibit precisely the desired features: they overlap substantially with each other, but not with the initial state.

We emphasize that the directions along which the momentum-space probability distribution is stretched versus squeezed are important. If one were to instead elongate it vertically, giving a substantial momentum uncertainty to one particle but not the other, the resulting performance would be very poor: there is no overlap between the two possible final states, and only scattering from one of the particles is reliably detectable. The fact that the interaction gives the same momentum transfer to whichever target/detector particle is involved singles out the $p_1 = p_2 \,(=p_3=...)$ direction, and thus it is crucial that one have little uncertainty along that direction and lots along orthogonal directions.

While the discussion so far has focused on processes that transfer momentum, the intuition for discrete and internal degrees of freedom is essentially identical. We illustrate this using the example of spin transfer to a target consisting of spin-1 particles in Table~\ref{tab:IntuitionTableDiscrete}. In each row, we show the initial distribution of spin states of each particle in the $z$ basis, i.e. indexed by the eigenvalue of $J_z$ for each particle; for simplicity we consider states consisting of equal-amplitude superpositions of these basis states, with no relevant phases between them. The plots thus show simply which basis states are included in the initial and final states. We assume the scattering process transfers one positive unit of spin along $z$ to either of the two particles.

The first three rows of Table~\ref{tab:IntuitionTableDiscrete} are essentially direct analogues of the matching rows in Table~\ref{tab:IntuitionTable}. Although the finite dimension of the Hilbert space results in some order-one changes to the behavior, it does not change any parametric scaling. We additionally include a fourth row in this case, however, illustrating the poor choice of ``direction'' for an initial state discussed above using a highly entangled GHZ-like state. Nonetheless, while entanglement is necessary for simultaneous coherence and detectability, it is not sufficient, and so the state shown here is not effective for detecting the chosen scattering process (although it could be effective for other measurements).

Finally, in Figure~\ref{fig:EntanglementAccumulation}, we illustrate how entanglement can accumulate over the integration time of a scattering detector. As we note in the main text, we are not aware of any situations in which this accumulation is sufficient to generate a parametric change between Lemma~\ref{lem:SingleParticleGeneralLimit} and Theorem~\ref{thm:MultipleParticleLimit}, but some degree of entangled accumulation is in fact very common. As a simple example of this, we consider the same hard scattering process as in Table~\ref{tab:IntuitionTable} (although a soft scattering process works as well; we use hard scattering solely for visual clarity). As Figure~\ref{fig:EntanglementAccumulation} shows, a series of repeated scatters without decoherence causes the target state to gradually become more and more elongated---and along precisely the direction needed to enhance scattering---as the scattering process leaves the uncertainty in $p_1 + p_2$ unchanged but generates uncertainty in $p_1 - p_2$. Notably, although we have illustrated this for a deterministic process with equal momentum transfers in the same direction each time, this is not necessary: a classical distribution of incident particles would lead to a mixed final state, but each of the pure state components of that final state would nonetheless become elongated in this fashion, leading to a similar accumulation of entanglement (now measured by $\EMF$). We show ane example of this in Figure~\ref{fig:EntanglementAccumulationMixed}.

An analogous accumulation can, of course, occur with discrete variables (e.g. spins), analogously to Table~\ref{tab:IntuitionTableDiscrete}. Just as in the momentum case above, these can often be understood as arising from a conserved quantity---in the spin case, angular momentum. Any particular set of outgoing states must correspond to a fixed amount of angular momentum being transferred to the target, but that angular momentum can be distributed across the target particles in many ways. This again leads to no increase in the uncertainty of e.g. $\sum J_z$ (within pure state components of the final state), while increasing the uncertainties in individual particles' spins, giving states similar to the third row of Table~\ref{tab:IntuitionTableDiscrete}.

\def\firstColWidth{0.24\textwidth}	

\newcommand{\splitcell}[2]{%
  \begin{tblr}{
    width=\linewidth,
    colspec={X[c,m]},
    rows={halign=c, valign=m, ht=6em},
    hline{2},
  }
  #1 \\
  #2
  \end{tblr}%
}

\begin{table}[h]
\centering

\begin{tblr}{
  width   = \textwidth,
  colspec = {Q[wd=\firstColWidth] Q[1] Q[1] Q[4]},
  cells = {halign=c, valign=m},
  hlines,
  vline{1,5} = {1-Z}{},
  vline{2-4} = {3-Z}{},
}
\SetCell[c=4]{c} {\textbf{Coherent and/or Detectable Scattering for $N=2$}} \\
\SetCell[c=4]{c} {%
  \legsq{figureRed}\; Before scattering\qquad
  \legsq{figureGreen}\; After scattering from particle 1\qquad
  \legsq{figureBlue}\; After scattering from particle 2
} \\
\textbf{Momentum Distributions} & \textbf{Rate} & \textbf{Trace Dist.} & \textbf{Notes} \\

\begin{minipage}[c]{\firstColWidth}\centering
\includegraphics[width=\linewidth]{HardScatter.pdf}
\end{minipage} & $O(N)$ & $O(1)$ & A typical example of an incoherent process, such as a hard scattering event. Since the momentum transfer is large compared to all momentum uncertainties of the initial state, the final states resulting from interactions with different particles/modes are essentially orthogonal, and thus the process is not coherent, giving an $O(N)$ scattering probability. The change in target state is, however, reliably detectable for any $N$. \\

\begin{minipage}[c]{\firstColWidth}\centering
\includegraphics[width=\linewidth]{SoftScatter.pdf}
\end{minipage} & \splitcell{$O(N^2)$}{$O(N)$} & \splitcell{$O\pfrac{1}{\sqrt{N}}$}{$O(1)$} & A poorly-detectable coherent process, in which the momentum transfer is small compared to all initial momentum uncertainties. While the total scattering probability can scale as $O(N^2)$, detection strategies using this full probability (e.g. by measuring $\sum p$) suffer from a trace distance that is $O(N^{-1/2})$ per scatter, limiting the experimental scaling to $O(N^{3/2})$. Other detection strategies (e.g. looking for large momenta outside of the initial state's support) may have $O(1)$ trace distances, but they rely on incoherent scattering contributions, with rates at only $O(N)$. \\

\begin{minipage}[c]{\firstColWidth}\centering
\includegraphics[width=\linewidth]{EntangledScatter.pdf}
\end{minipage} & $O(N^2)$ & $O(1)$ & An example of simultaneous coherence and high detectability; note the necessary entanglement between the two particles' momenta. The small uncertainty in the total (i.e. center-of-mass) momentum makes the state change reliably detectable at any $N$, but the large uncertainty in the orthogonal momentum direction (i.e. the difference of momenta) makes the different final states highly overlapping, permitting an $O(N^2)$ scattering probability. \\

\end{tblr}

\caption{An illustration of scattering/interaction processes that transfer momentum, illustrating various combinations of coherence and detectability. Here we use ``coherence'' to refer to interaction probabilities/rates that scale as the number of particles squared, and ``detectability'' to refer to processes where the initial and final states are separated by an $N$-independent trace distance, such that an optimal measurement can distinguish them with non-negligible probability in the large-$N$ limit. In each row, the plot qualitatively illustrates the initial (red) probability distribution of each particle's momentum, as well as the distributions after the scatter/interaction happens with the first/second (green/blue) particle; the final state (in the absence of decoherence) will thus be a superposition of the green and blue states. The arrows in each figure give the corresponding momentum transfer, whose size can be compared against the characteristic momentum uncertainties of the various initial states.}
\label{tab:IntuitionTable}
\end{table}

\vspace{-1cm}

\begin{table}[h]
\centering

\begin{tblr}{
  width   = \textwidth,
  colspec = {Q[wd=\firstColWidth] Q[1] Q[1] Q[4]},
  cells = {halign=c, valign=m},
  hlines,
  vline{1,5} = {1-Z}{},
  vline{2-4} = {3-Z}{},
}
\SetCell[c=4]{c} {\textbf{Coherent and/or Detectable Inelastic Scattering from $N=2$ Spin-1 Particles}} \\
\SetCell[c=4]{c} {%
  \legsq{figureRed}\; Before scattering\qquad
  \legsq{figureGreen}\; After scattering from particle 1\qquad
  \legsq{figureBlue}\; After scattering from particle 2
} \\
\textbf{Spin Distributions} & \textbf{Rate} & \textbf{Trace Dist.} & \textbf{Notes} \\

\begin{minipage}[c]{\firstColWidth}\centering
\includegraphics[width=\linewidth]{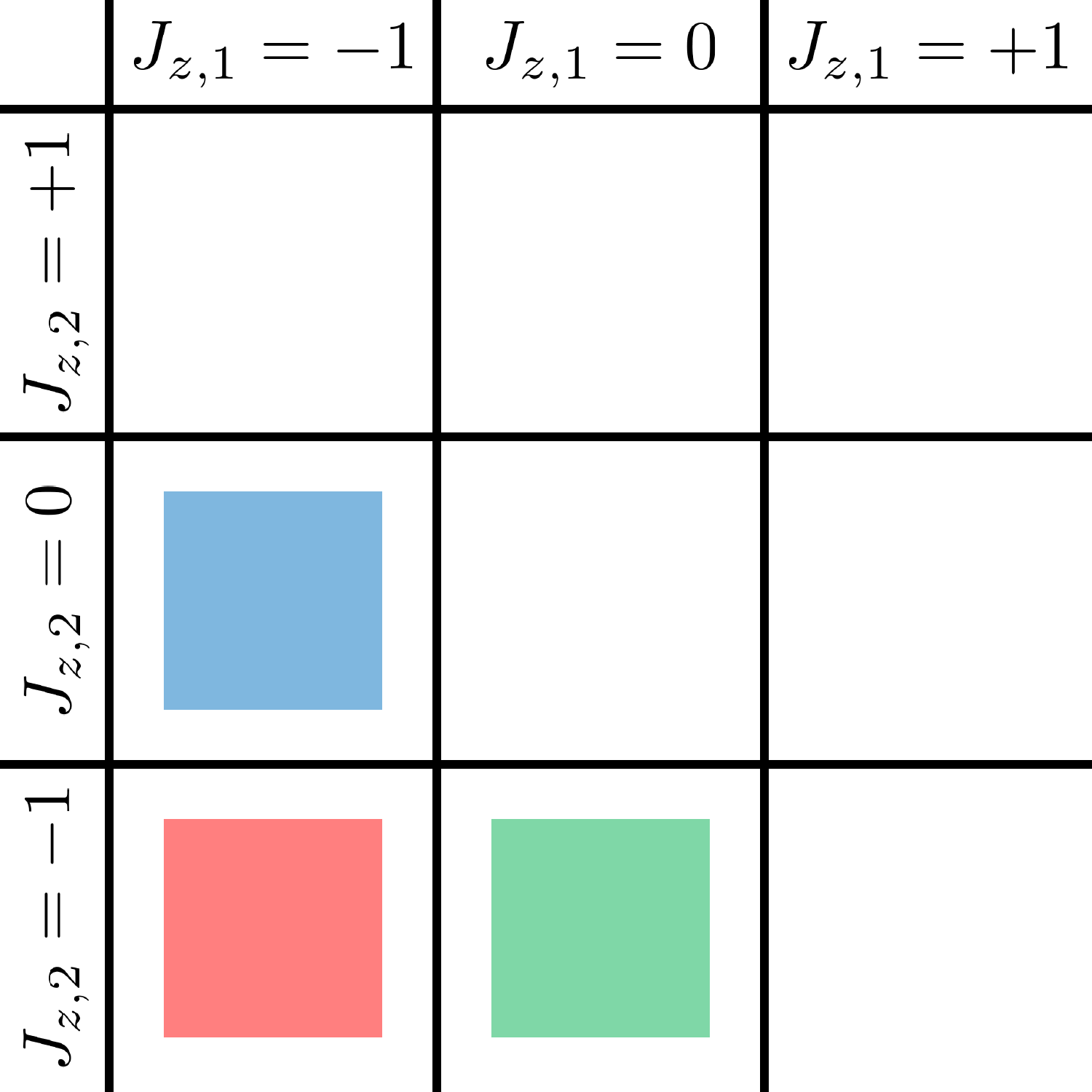}
\end{minipage} & $O(N)$ & $O(1)$ & A simple example of an incoherent inelastic scattering process from the $\ket{-1}^{\otimes2}$ state (i.e. both target spins initially in $J_z = -1$). Since both particles are in spin eigenstates, the two possible final states are orthogonal both to the initial state and to each other, preventing coherence. Thus, while the initial and final states are reliably distinguished at any $N$, the scattering rate scales only as $O(N)$. \\

\begin{minipage}[c]{\firstColWidth}\centering
\includegraphics[width=\linewidth]{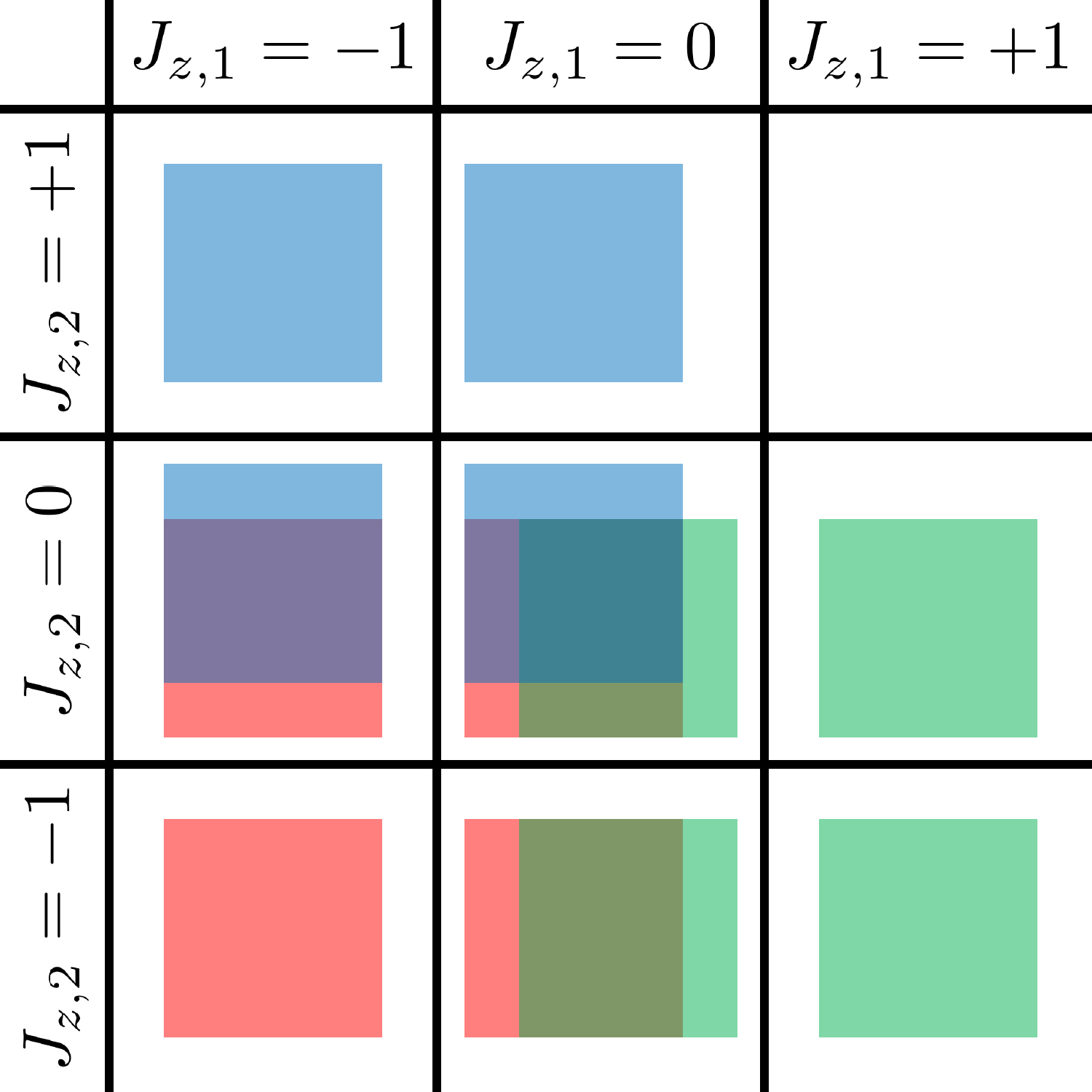}
\end{minipage} & \splitcell{$O(N^2)$}{$O(N)$} & \splitcell{$O\pfrac{1}{\sqrt{N}}$}{$O(1)$} & An example of a (partially) coherent process, in which scattering from a product state of spin superpositions, $(\ket{-1}+\ket{0})^{\otimes 2}/2$, leads to final states that have incomplete overlap with both the initial state and with each other, resulting in partial coherence and partial detectability. This is loosely analogous to the second row of Table~\ref{tab:IntuitionTable}, but the finite dimension of this Hilbert space reduces the final state overlaps. Nonetheless it has the same parametric scaling: an $O(N^2)$ scattering rate that contributes $O(N^{-1/2})$ trace distance, with $O(N)$ scattering rate into states that differ by $O(1)$. \\

\begin{minipage}[c]{\firstColWidth}\centering
\includegraphics[width=\linewidth]{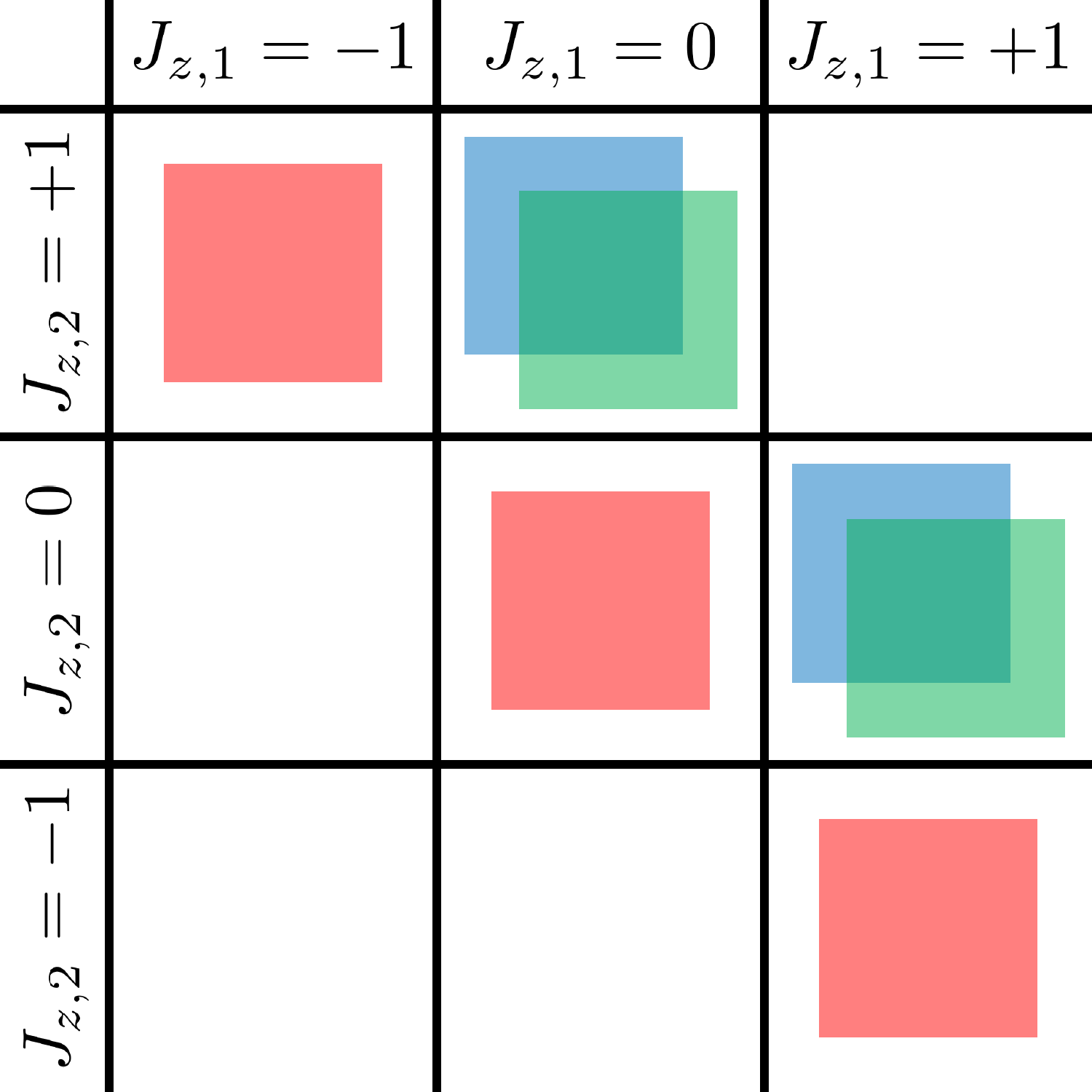}
\end{minipage} & $O(N^2)$ & $O(1)$ & An example of scattering that is both coherent and detectable, thanks to the entanglement of the initial state, $(\ket{+1}\ket{-1} + \ket{0}^{\otimes2} + \ket{-1}\ket{+1})/\sqrt{3}$. This is a straightforward, discrete analogue of the third row in Table~\ref{tab:IntuitionTable}, although the finite dimension of the spin Hilbert space changes the details of the behavior. \\

\begin{minipage}[c]{\firstColWidth}\centering
\includegraphics[width=\linewidth]{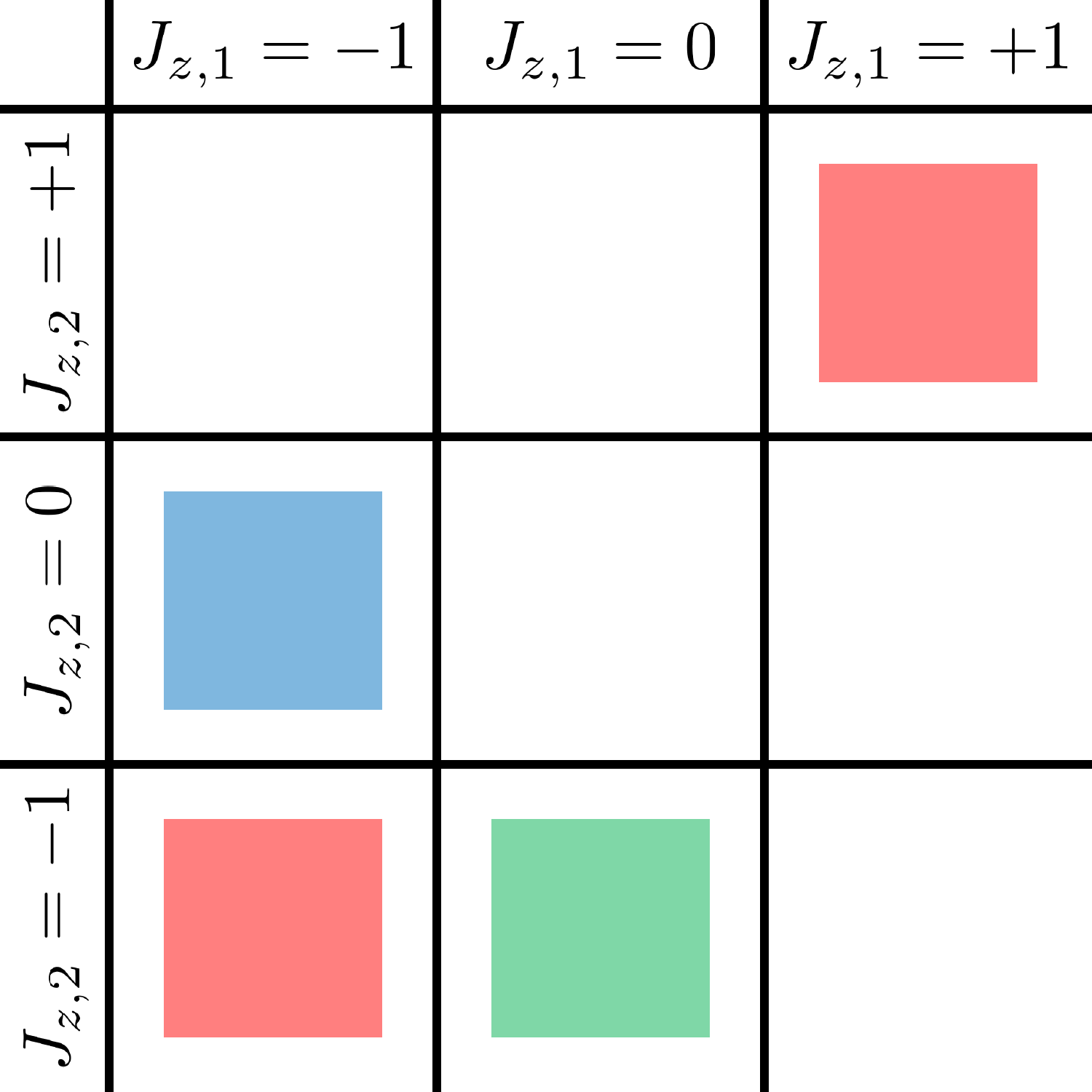}
\end{minipage}
& $O(N)$ & $O(1)$ & An illustration that, while entanglement is necessary for detectable coherent scattering, it is not sufficient: Although the GHZ-like initial state shown here, $(\ket{-1}^{\otimes 2} + \ket{+1}^{\otimes 2})/\sqrt{2}$, is highly entangled, it is not well-suited to the scattering process considered in this table, and therefore performs poorly. This state can, however, be effective for other measurements (e.g. of a Hamiltonian $H \propto J_{z,1}+J_{z,2}$). \\

\end{tblr}

\caption{An illustration of scattering/interaction processes that transfer spin along the $z$ axis to a target composed of spin-1 particles, illustrating various combinations of coherence and detectability. Here we use ``coherence'' to refer to interaction probabilities/rates that scale as the number of particles squared, and ``detectability'' to refer to processes where the initial and final states are separated by an $N$-independent trace distance, such that an optimal measurement can distinguish them with non-negligible probability in the large-$N$ limit. In each row, the plot qualitatively illustrates the initial (red) probability distribution of each particle's spin component $J_z$, as well as the distributions after the scatter/interaction happens with the first/second (green/blue) particle; the final state (in the absence of decoherence) will thus be a superposition of the green and blue states. The spin change is fixed to be $+1$ throughout, i.e. either the first or the second particle's $J_z$ increases by $1$ after the interaction.}
\label{tab:IntuitionTableDiscrete}
\end{table}

\begin{figure}[h]
    \centering
    \includegraphics[width=0.3\textwidth]{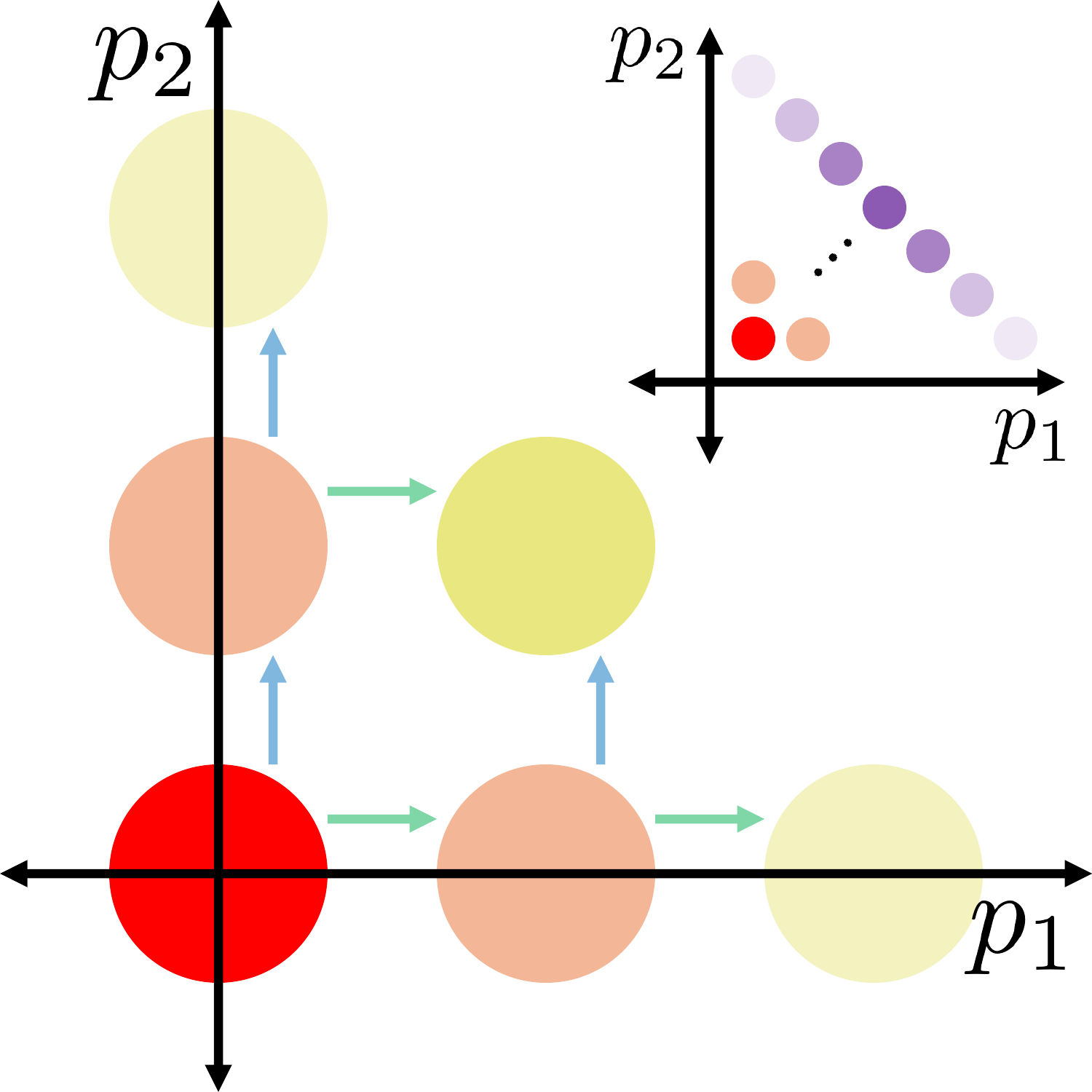}
    \caption{An illustration of how a sequence of scattering events can lead to entanglement in a target. As in Table~\ref{tab:IntuitionTable}, we consider scattering processes that transfer some fixed momentum to either of two particles in a target, whose combined momentum distribution is plotted. The initial state is red, while momentum transfers to the first/second particle are illustrated in green/blue. Thus, after one scattering---and in the absence of any decoherence---the final state is the orange state, in an even superposition of either of the two particles having the additional momentum. After a second scattering, the target reaches the yellow state. In the inset, we show how, after many scattering events, the violet final state approaches a band of nearly fixed $p_1 + p_2$ but widely varying $p_1 - p_2$; this is exactly the form of entanglement that can lead to coherent and distinguishable scattering, as shown in the third row of Table~\ref{tab:IntuitionTable}. Note that, although for simplicity we illustrate a sequence of scattering events whose momentum transfers have uniform signs and magnitudes, this is not necessary: a random distribution of scattering events will lead to a (classical) mixture of different final $p_1 + p_2$ values, but each of the pure state components of the mixed final state will still be entangled like the state illustrated here, and coherent distinguishable scattering from the final state will still be possible. We show an example of this in Figure~\ref{fig:EntanglementAccumulationMixed}.}
    \label{fig:EntanglementAccumulation}
\end{figure}

\begin{figure}[h]
    \includegraphics[width=0.99\linewidth]{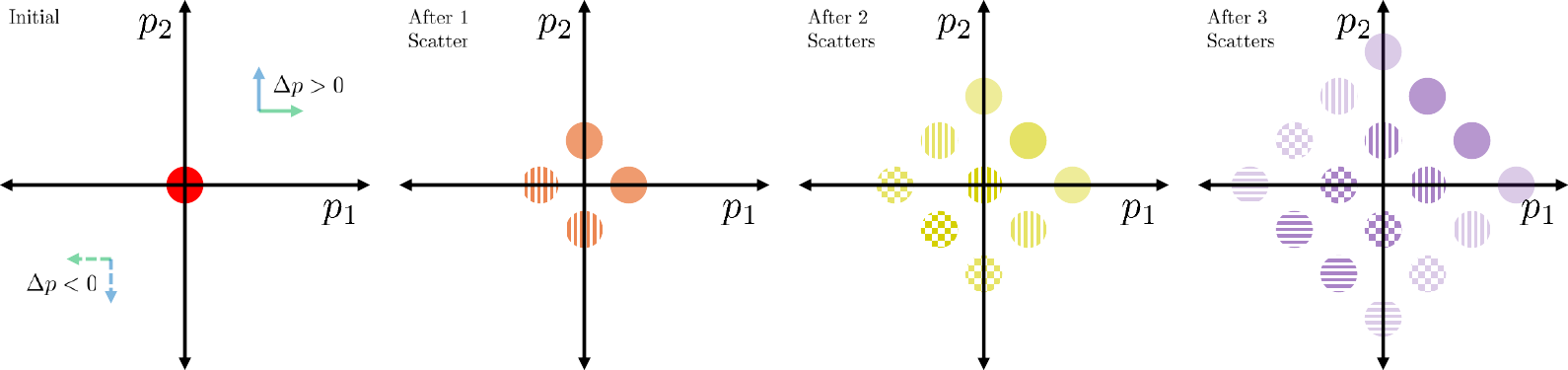}
    \caption{An example of how entanglement accumulation via scattering can occur even when successive scattering events are not identical. In this case we consider a momentum transfer to either particle, analogous to Figure~\ref{fig:EntanglementAccumulation}, but now allowing either sign of the momentum transfer. Since the outgoing particle state is different depending on the momentum transfer's sign, this leads to target states that are mixed, with entanglement within each pure state component of that mixed state. Thus, after one scatter (in orange), the target is in a mixed state composed of two copies of the orange state in Figure~\ref{fig:EntanglementAccumulation} (with one shifted); we show one solid and one striped to emphasize that the two are classically mixed, rather than in a superposition (although each is itself a superposition of two states). Likewise there are three (four) pure state components to the final mixed states after two (three) scatters, each one entangled, shown in yellow (purple) with different patterns for each component. The state's $\EMF$ thus grows over time.}
    \label{fig:EntanglementAccumulationMixed}
\end{figure}

\end{document}